\documentclass[12pt,a4paper]{amsproc}

\usepackage[margin=1.5cm]{geometry}
\usepackage{amsmath,amssymb,amsthm}
\usepackage{mathrsfs}
\usepackage{enumitem}
\usepackage{graphicx}
\usepackage{color}
\usepackage{diagbox}

\newtheorem{theorem}{Theorem}

\newtheorem{lemma}{Lemma}
\newtheorem{proposition}{Proposition}
\theoremstyle{definition}

\theoremstyle{remark}
\newtheorem{remark}{Remark}

\def\R{\mathbb{R}}
\def\1{\mathbf{1}}
\def\ve{\varepsilon}
\def\maxZ1{Z^1_{\text{max}}}
\def\Femp{\mathbb{F}}
\def\hFemp{\hat{\Femp}}

\def\hM{\hat{M}}
\def\hQ{\hat{Q}}

\def\Haz{\Lambda}
\def\hHaz{\hat{\Haz}}
\def\Var{\mathrm{Var}}
\def\op{o_P(1)}

\def\opn{o_P(n^{-1/2})}
\def\Opn{O_P(n^{-1/2})}
\def\sj{\sum_{j=1}^n}
\def\avj{\frac1n\sj}
\def\sk{\sum_{k=1}^n}

\def\hxi{\hat{\xi}}
\def\mhat{\hat{m}}
\def\shat{\hat{s}}
\def\hpi{\hat{\pi}}

\numberwithin{equation}{section}

\allowdisplaybreaks

\begin{document}

\title[The nonparametric mixture cure model]
{The nonparametric location--scale mixture cure model}

\author[J.\ Chown]{Justin Chown$^*$}
\address{Fakult\"at f\"ur Mathematik\\
Lehrstuhl f\"ur Stochastik\\
Ruhr-Universit\"at Bochum\\
44780 Bochum\\
Germany}
\email{justin.chown@ruhr-uni-bochum.de}
\thanks{$^*$Correspondences may be addressed to Justin Chown.}

\author[C.\ Heuchenne]{C\'edric Heuchenne}
\address{HEC Li\'ege, University of Li\'ege\\
Institute of Statistics, Biostatistics and Actuarial Sciences\\
Universit\'e catholique de Louvain\\
Rue Louvrex 14\\
4000 Li\'ege\\
Belgium}
\email{c.heuchenne@uliege.be}

\author[I.\ Van Keilegom]{Ingrid Van Keilegom}
\address{ORSTAT\\
KU Leuven\\
Naamsestraat 69\\
B 3000 Leuven\\
Belgium}
\email{ingrid.vankeilegom@kuleuven.be}

\begin{abstract}
We propose completely nonparametric methodology to investigate
location--scale modelling of two--component mixture cure models, where
the responses of interest are only indirectly observable due to the
presence of censoring and the presence of so--called {\em long--term
  survivors} that are always censored. We use covariate-localized
nonparametric estimators, which depend on a bandwidth sequence, to
propose an estimator of the error distribution function that has not
been considered before in the literature. When this bandwidth belongs
to a certain range of undersmoothing bandwidths, the asymptotic
distribution of the proposed estimator of the error distribution
function does not depend on this bandwidth, and this estimator is
shown to be root-$n$ consistent. This suggests that a computationally
costly bandwidth selection procedure is unnecessary to obtain an
effective estimator of the error distribution, and that a simpler
rule-of-thumb approach can be used instead. A simulation study
investigates the finite sample properties of our approach, and the
methodology is illustrated using data obtained to study the behavior
of distant metastasis in lymph-node-negative breast cancer patients.
\end{abstract}
\bigskip

\maketitle

\bigskip


\section{Introduction}
\label{intro}

A common problem faced in medical studies is for some subjects to
never experience the event of interest during the study period. For
example, consider a follow--up study examining the harmful
side--effects of a pharmaceutical product. Since side--effects are
commonly rare, it is expected that many subjects involved in the study
will not experience the harmful effect of the treatment by the end of
the study, and, therefore, these subjects will be censored at the
conclusion of the study. Hence, these subjects are called {\em
  long--term survivors}. The first known study involving statistical
analysis of data containing long--term survivors dates back to Boag
(1949), and this author coined the term {\em cure model} to indicate
the data contained a non--trivial proportion of long--term
survivors. However, Farewell (1986) observes that cure models should
not be used without clear empirical or biological need (see Section
\ref{estimates}). Results on censored data models should be
used in these cases, and studies involving subjects with censored
responses are common. Well known methods can be employed to study
these data (see, for example, Lawless, 1982; Aitkin et al.,1989; Harris
and Albert, 1991; Collett, 1994).

We consider the case of observing responses $Y$ that are right
censored by another random variable $C$, and, hence, observe only the
minimum $Z = Y \wedge C$. Throughout this article, we will assume that
$Y$ and $C$ are only conditionally independent given a covariate $X$,
and, for simplicity, we will assume the censoring variable $C$ is a
continuous random variable. For clarity, we will refer to responses
corresponding to the subpopulation that has survival times that are
finite as $Y_u$, which are only {\em indirectly observed} due to the
presence of censored values. The purpose of this paper is to study the
heteroskedastic nonparametric regression of $Y_u$ given the covariate
$X$:
\begin{equation} \label{modeleq}
Y_u = m(X) + s(X)\ve.
\end{equation}
Here $m$ is the regression function and $s$ is the scale function
(bounded away from zero), which are both assumed to be smooth. We
assume the error $\ve$ is a continuous random variable that is
independent of the covariate $X$ and has distribution function
$F$. Identifiability of the components in the cure model also requires
additional assumptions on the joint distribution of $(X,Z,\delta)$,
where $\delta = \1[Y \leq C]$ is the right--censoring
indicator. Finally, $m$ is assumed to be a location--type functional
and $s$ is assumed to be a scale--type functional, which imposes
additional requirements on the model error $\ve$ that are analogous to
the usual zero mean and unit variance assumptions so that we can
identify both of $m$ and $s$ (see Section \ref{estimates}).

From the random sample of data
$(X_1,Z_1,\delta_1),\ldots,(X_n,Z_n,\delta_n)$ we will propose
nonparametric function estimators of $m$, $s$ and $F$ in Section
\ref{estimates}, and we study the large sample behavior of these
estimators in Section \ref{asymptotics}. Note, these data also include
the cured cases and, therefore, do not directly correspond with
\eqref{modeleq}. Estimating the error distribution function $F$ is
particularly important because many statistical inference procedures
depend on functionals of $F$; for example, Kolmogorov--Smirnov--type
and Cram\'er--von--Mises--type statistics. Van Keilegom and Akritas
(1999) considered estimation of $F$, as well as the functions $m$ and
$s$, from a model similar to \eqref{modeleq} but without considering
cured subjects, which presents a new and important challenge, and,
hence, those results do not directly apply to the present situation.

We are interested in studying a completely nonparametric statistical
methodology for examining location--scale modelling of data containing
long--term survivors. Since the study of this type of data presented
in Boag (1949), the literature on this subject has been divided into
two distinct categories. The original cure model proposed by Boag
(1949) is now known as the two--component mixture cure model (see, for
example, Taylor, 1995). We will refer to the second category as simply
the non--mixture cure model (see, for example, Haybittle, 1959,
1965). Several advancements have been made in the literature when
these models are assumed to have a parametric or a semiparametric
form. Kuk and Chen (1992) investigate combining logistic regression
methods (used to estimate the unknown proportion of cured cases) with
proportional hazards techniques to obtain estimators of their model
parameters. Taylor (1995) works with a similar model as that of Kuk and
Chen (1992) but proposes an Expectation-Maximization algorithm for
simultaneously fitting the model parameters and estimating the
baseline hazard function, and this author also uses simulations to
conjecture that a crucial assumption is required for identifying
components of these models (see Section \ref{estimates}). Sy and
Taylor (2000) consider a similar model as that of Kuk and Chen (1992)
but investigate an estimation technique based on the
Expectation-Maximization algorithm. Lu (2008) considers the
proportional hazards mixture cure model and proposes a semiparametric
estimator of the unknown parameters of that model by maximizing a
so-called nonparametric likelihood function. The estimator is shown to
have minimum asymptotic variance. Lu (2010), the same author as
before, investigates an accelerated failure time cure model (a special
case of the two--component mixture cure model) and proposes an
Expectation-Maximization algorithm for fitting the unknown parameters
of that model as well as obtaining an estimator of the unknown error
density function using a kernel-smoothed conditional profile
likelihood technique. Xu and Peng (2014) and L\'opez-Cheda et al.\
(2017) consider estimating the unknown cure fraction in a completely
nonparametric setting. Patilea and Van Keilegom (2017) consider a
general approach to modelling the conditional survival function of a
subject given that the subject is not cured by proposing so-called
inversion formulae that allows one to express the conditional survival
function of the uncured subjects in terms of the proportion of cured
subjects and the subdistributions of the response and the censoring
variable.

A particularly interesting model belonging to the non--mixture case is
proposed by Yakovlev and Tsodikov (1996). Earlier, in a back-and-forth
exchange over letters to the editors of {\em Statistics in Medicine},
Andrej Yakovlev very clearly details shortcomings of the
two--component mixture cure model. Specifically, he argues that the
two--component mixture cure model implicitly assumes that there is
only a single risk operative in the population, i.e.\ a single
latent variable that determines whether or not subjects are cured /
uncured. Yakovlev argues that, in general, populations have multiple
risk operatives and he proposes what we are referring to as the
non--mixture cure model in his letter to the editors. The exchange
between him and the authors Alan Cantor and Jonathan Shuster can be
found in Yakovlev et al.\ (1994). Sometimes this model is called a
promotion time cure model, and it has gained popularity due to
motivations from Biology (in particular cancer research). Several
results on parametric and semiparametric estimation of the unknown
parameters in these models are available in the literature. However,
this model is not in the scope of the current article, and we only
mention a few of the notable works in this area. Tsodikov (1998)
compares likelihood--based fitting techniques for the non--mixture
cure model. Later, Tsodikov et al.\ (2003) survey the
literature and report that, in less than a decade from its
introduction, the non--mixture cure model is already in popular use,
and these authors promote use of the non--mixture cure model in
semiparametric and Bayesian settings. Zeng et al.\ (2006) propose a
recursive algorithm for obtaining maximum likelihood estimates in the
non--mixture cure model. Additionally, these authors show their
regression parameter estimators have minimum asymptotic
variance. Portier et al.\ (2017) consider an extension of the
non--mixture cure model.

Recently, efforts have been made to unify the two seemingly distinct
categories of cure models. Sinha et al.\ (2003) discuss the
benefits and disadvantages of the mixture and non--mixture cure
models. Yin and Ibrahim (2005) propose transforming the unknown
population survival function in a manner that is analogous to the
Box-Cox transformation for non-normally distributed random
variables.

Our goal is to generally relax the model conditions and investigate an
alternative estimation strategy for the two--component mixture cure
model. A popular theme in both modelling categories is to use
proportional hazards methods, where one estimates the baseline hazard
function using a nonparametric estimator. In this case, the
proportional hazards model is made additionally flexible through a
nonparametric estimator of the baseline hazard function. The result is
a semiparametric estimation technique (the remaining model parameters
are finite dimensional and a likelihood function is usually required
to obtain estimators).

The article is organized as follows. We further discuss model
\eqref{modeleq} and motivate our nonparametric function estimators in
Section \ref{estimates}, and we give asymptotic results on these
estimators in Section \ref{asymptotics}. In Section
\ref{applications}, we investigate the finite sample properties of our
proposed estimator of $F$ and we illustrate the proposed techniques by
characterizing a set of data collected to study distant metastasis in
lymph-node-negative breast cancer patients. Our numerical study of the
previous results in Section \ref{simulations} shows the finite sample
behavior of the estimators proposed in Section \ref{estimates} is
well--described by the asymptotic statements given in Section
\ref{asymptotics}. The proofs of these results and further supporting
technical results are given in Section \ref{appendix}.


\section{Estimation of the model parameters}
\label{estimates}
We begin this section with a discussion of the identifiability of the
cure model parameters. In the following, write $G$ for the
distribution function of the covariates $X$ and $g$ for the density
function of $G$, where the support of $X$ is $[0,\,1]$. Let $Q$ be the
conditional distribution function of the responses $Y$ given $X$ and
$Q_u$ be the conditional distribution function of $Y_u$ given $X$. For
both cure models and censored response models, it is important that we
place conditions on the distribution function $Q$ (and therefore
$Q_u$) so that we may identify and estimate the regression model
components $m$ and $s$ and the error distribution function
$F$. Empirical or biological need for using a cure model in the
present situation means the support of the censoring variable $C$ is
never contained in the support of the subpopulation $Y_u$, i.e.\ we
require
\begin{equation} \label{curemodid}
\tau_0 = \sup_{x \in [0,\,1]} \tau_u(x) < \tau_C(x),
 \qquad x \in [0,\,1],
\end{equation}
where $\tau_u(x) = \inf \{t \in \R \,:\, 1 - Q_u(t\,|\,x) = 0\}$
and $\tau_C(x) = \inf \{t \in \R \,:\, P(C > t\,|\,X = x) =
0\}$. Taylor (1995) uses simulation evidence to conjecture the
necessity of \eqref{curemodid}. Xu and Peng (2014) and L\'opez-Cheda
et al.\ (2017) observe that \eqref{curemodid} leads to identifiability
of the cure model components (see Lemma 1 of L\'opez-Cheda et al.,
2017); specifically, it is required to identify the conditional
proportion $\pi(X)$ of cured cases given $X$ as well as the
conditional distribution function $Q_u$ of $Y_u$ given $X$. To ensure
the distribution of the censoring variable $C$ is identifiable, we
will further assume that the remaining mass of $Y$ beyond $\tau_u(X)$
occurs at $Y = \infty$, i.e.\ we assume the conditional equivalence of
the events $\{Y > t\} = \{Y = \infty\}$, $t \geq \tau_u(X)$, given
$X$. This justifies writing
\begin{equation*}
P(Y > t\,|\,X) = \pi(X) + \{1 - \pi(X)\}P(Y_u > t\,|\,X),
 \qquad t \in \R,
\end{equation*}
where $\pi(X) = P(Y > \tau_u(X)\,|\,X) = P(Y = \infty \,|\,X)$ is
assumed to be bounded away from zero and one, i.e.\ there are finite
positive real numbers $0 < \pi_l \leq \pi_u < 1$ satisfying $\pi_l
\leq \pi(X) \leq \pi_u$ for every $X$. Hence, $P(Y > \tau_u(X)\,|\,X) =
\pi(X) = P(Y > \tau_0\,|\,X)$. This means that \eqref{curemodid}
implies that we only need an estimator of $\tau_0$, which does not
depend on $X$, rather than $\tau_u(\cdot)$.

To conclude our discussion on identifiability, recall that the
regression function $m$ is a location--type functional and the scale
function $s$ is a scale--type functional. This means there are
transformations $T$ and $V$ such that
\begin{equation*}
m = T\big(Q_u) = m + sT(F)
\quad\text{and}\quad
s = V(Q_u) = sV(F).
\end{equation*}
Therefore, we can see that the regression model components $m$ and $s$
are identifiable when $T(F) = 0$ and $V(F) = 1$.

As noted on page 186 in Dabrowska (1987), responses arising from
experiments with censored values are often skewed to the right and,
therefore, estimators of the mean (and scale) will be
affected. Beran (1981) proposes using $L$--type regression functionals,
which are more robust to skewing in the data. To explain the idea, we
introduce the score function $J$ and the quantiles $\xi_u(p\,|\,x) =
Q_u^{-1}(p\,|\,x) := \inf \{y \in (-\infty,\,\tau_0]\,:\,Q_u(y\,|\,x)
\geq p\}$ for $p \in [0,\,1]$. Here the score function $J$ must be
nonnegative and satisfy $\int_0^1 J(p)\,dp = 1$. Throughout this
article, we work with the following definitions of $m$ and $s$:
\begin{equation*}
m(x) = \int_0^1 \xi_u(p\,|\,x)J(p)\,dp
\quad\text{and}\quad
v(x) = \int_0^1 \xi_u^2(p\,|\,x)J(p)\,dp - m^2(x),
\end{equation*}
where $s(x) = v^{1/2}(x)$ and $x \in [0,\,1]$. Hence, for $m$ and $s$
to be identifiable, we will require that $F$ satisfies
\begin{equation*}
\int_0^1 \xi_F(p)J(p)\,dp = 0
\quad\text{and}\quad
\int_0^1 \xi_F^2(p)J(p)\,dp = 1,
\end{equation*}
where $\xi_F$ is the quantile function of $F$, i.e.\ $\xi_F(p) =
\inf\{t \in \R\,:\,F(t) \geq p\}$ for $p \in [0,\,1]$.

With all of the components of the regression model \eqref{modeleq}
identified, we can define our estimators of the model parameters. To
define the estimator of $Q$, we will introduce further notation. Write
$M$ for the conditional distribution function of the minimum $Z$ given
$X$ and $M^1$ for the conditional subdistribution function of both $Z$
and $\delta = 1$ given $X$. From the discussion above, we can see that
$\tau_M(x) = \inf\{t \in \R\,|\,P(Z > t\,|\,X = x) = 0\} = \tau_C(x) >
\tau_0$ for every $x$, which follows by \eqref{curemodid}, and, hence,
we can consistently estimate $Q$ (and therefore $Q_u$) everywhere on
the region $(-\infty,\,\tau_0] \times [0,\,1]$, cf.\ Van Keilegom and
Akritas (1999).

Using $M$ and $M^1$, the conditional cumulative hazard function $\Haz$
of $Y$ given $X$ may be written as
\begin{equation} \label{hazY}
\Haz(t\,|\,X) = \int_{-\infty}^t \frac{Q(ds\,|\,X)}{1 - Q(s-\,|\,X)}
 = \int_{-\infty}^t \frac{M^1(ds\,|\,X)}{1 - M(s-\,|\,X)},
 \qquad t \in (-\infty, \tau_0].
\end{equation}
To estimate $M$ and $M^1$, we introduce the Nadaraya--Watson weights
\begin{equation*}
W_j(x) = K\bigg(\frac{x - X_j}{a_n}\bigg) \Bigg\slash
 \Bigg\{ \sk K\bigg(\frac{x - X_k}{a_n}\bigg) \Bigg\},
 \qquad j = 1,\ldots,n,
\end{equation*}
where $K$ is a given kernel function and $\{a_n\}_{n \geq 1}$ is a
sequence of bandwidth parameters. Later, we will specify conditions on
choosing $K$ and $\{a_n\}_{n \geq 1}$. Estimates of $M$ and $M^1$ then
follow by the approach of Stone (1977): for $(t,\,x) \in
(-\infty,\,\tau_0] \times [0,\,1]$,
\begin{equation} \label{hMhM1}
\hM(t\,|\,x) = \sj \1(Z_j \leq t)W_j(x)
\quad \text{and} \quad
\hM^1(t\,|\,x) = \sj \delta_j\1(Z_j \leq t)W_j(x).
\end{equation}
Substituting \eqref{hMhM1} into \eqref{hazY} leads to an estimator of
$Q$ in the approach of Beran (1981):
\begin{equation} \label{hQ}
\hQ(t\,|\,x) = 1 - \prod_{Z_{(j)} < t}
 \bigg\{ 1 - \frac{W_{(j)}(x)}{\sum_{k=j}^n W_{(k)}(x)}
 \bigg\}^{\delta_{(j)}},
 \qquad (t,\,x) \in (-\infty,\,\tau_0] \times [0,\,1].
\end{equation}
Here $Z_{(1)} \leq \ldots \leq Z_{(n)}$ is the ascending ordering of
$Z_1,\ldots,Z_n$ and both of $\delta_{(1)},\ldots,\delta_{(n)}$ and
$W_{(1)}(x),\ldots,W_{(n)}(x)$ are ordered according to
$Z_{(1)},\ldots,Z_{(n)}$. For simplicity we will assume that the data
contain no tied responses, which is reasonable because our
assumptions imply the responses $Z_j$ are continuous random
variables. Otherwise the ordering of the variables indicated above is
not unique and the estimator $\hQ$ is affected.

Xu and Peng (2014) propose estimating $\tau_0$ by the largest
uncensored response $\maxZ1$ and then combining this estimator with
\eqref{hQ} to form an estimator the unknown proportion $\pi$ of cured
cases,
\begin{equation} \label{pihat}
\hpi(x) = 1 - \hQ(\maxZ1 \,|\, x),
 \qquad x \in [0,\,1].
\end{equation}
The estimator $\hpi$ is shown to be consistent and asymptotically
normally distributed. Later, L\'opez--Cheda et al.\ (2017) generalize
this result in two steps. First, these authors show the estimator
$\maxZ1$ is strongly consistent for $\tau_0$. Second, the estimator
$\hpi$ is shown to be a uniformly, strongly consistent estimator of
$\pi$.

Turning our attention now to $m$ and $s$, we can see that the unknown
quantiles $\xi_u$ of the {\em uncured population} must be
estimated. It is easy to show the equivalence $\xi_u(p\,|\,x) = \xi((1
- \pi(x))p\,|\,x)$ from the equivalence $Q(\cdot\,|\,x) = \{1 -
\pi(x)\}Q_u(\cdot\,|\,x)$, with $p \in [0,\,1]$ and $x \in [0,\,1]$,
where $\xi((1 - \pi(x))p\,|\,x) = \inf\{y \in
(-\infty,\,\tau_0]\,:\,Q(y\,|\,x) \geq \{1 - \pi(x)\}p\}$. We can
consistently estimate $\xi((1 - \pi(x))p\,|\,x)$ by $\hxi((1 -
\hpi(x))p\,|\,x)$, where $\hxi((1 - \hpi(x))p\,|\,x) = \inf\{y \in
(-\infty,\,\tau_0]\,:\,\hQ(y\,|\,x) \geq \{1 - \hpi(x)\}p\}$. The
resulting estimators of $m$ and $s$ are analogous to those of Van
Keilegom and Akritas (1999):
\begin{equation*}
\mhat(x) = \int_0^1 \hxi\big(\big(1 - \hpi(x)\big)p\,|\,x\big)
 J(p)\,dp
\quad\text{and}\quad
\hat v(x) = \int_0^1 \hxi^2\big(\big(1 - \hpi(x)\big)p\,|\,x\big)
 J(p)\,dp - \mhat^2(x),
\end{equation*}
with $\shat(x) = \hat{v}^{1/2}(x)$, $x \in [0,\,1]$.

Write $\tau_F = \inf\{t \in \R:\, 1 - F(t) = 0\}$ for the largest
observable standardized error, which is finite by
\eqref{curemodid}. It follows that $\{\tau_u(X) -
m(X)\}/s(X)$ does not depend on $X$ and $\tau_F = \{\tau_u(X) -
m(X)\}/s(X)$, for $G$--almost every $X$, from
standardization. This means we can transfer the support of $F$,
$(-\infty,\,\tau_F]$, into the support of $Q$, $(-\infty,\,\tau_u(x)]
= (-\infty,\,\tau_Fs(x) + m(x)]$, $x \in [0,\,1]$, where $Q$ can be
estimated. Note, this is the same transfer of information from $F$ to
$Q$ studied in Van Keilegom and Akritas (1999). However, this implies
that we can form an estimator of $F$ using the estimators of $Q$,
$\pi$, $m$ and $s$, which is new.

Observe the error distribution function $F$ can be written as the
average
\begin{equation*}
F(t) = E\bigg[\frac{Q(t s(X) + m(X)\,|\,X)}{
 1 - \pi(X)}\bigg],
 \qquad -\infty < t \leq \tau_F,
\end{equation*}
where we have used the transference mapping $t \mapsto ts(x) + m(x)$
for $-\infty < t \leq \tau_F$. We arrive at the proposed estimator of
$F$:
\begin{equation} \label{Fhat}
\hFemp(t) = \avj \frac{\hQ(t\shat(X_j) + \mhat(X_j)\,|\,X_j)}{
 1 - \hpi(X_j)},
\qquad -\infty < t \leq \tau_F.
\end{equation}
Note this estimator is averaging over the local model estimators at each
covariate $X_j$ that are not consistent at the root-$n$ rate but are
consistent at slower rates. Nevertheless, we show the estimator
$\hFemp$ is root-$n$ consistent for $F$ and satisfies a functional
central limit theorem (see Section \ref{asymptotics}).


\subsection{Asymptotic results on the nonparametric function estimators}
\label{asymptotics}

In order to state our asymptotic results for the estimators introduced
in the previous section, we require the following assumptions:

\begin{enumerate}[label=\textnormal{(A\arabic*)}]
\item The bandwidth $a_n$ satisfies
  $(na_n^2)^{-1}\log\log(n) = O(1)$ and $na_n^5\log^{-1}(n) = O(1)$.
 \label{assumpbw}
\item There are real numbers $0 < \pi_l \leq \pi_u < 1$ satisfying
  $\pi_l < \pi(X) < \pi_u$ for every $X$.
 \label{assumppi}
\item
  \begin{itemize}
  \item[(i)] The kernel function $K$ is a symmetric probability
    density function with support $[-1,\,1]$.
  \item[(ii)] $K$ has bounded second derivative.
  \end{itemize}
  \label{assumpK}
\item 
  \begin{itemize}
  \item[(i)] The distribution function $G$ of the covariates $X$ has a
    density function $g$ that is bounded and bounded away from zero in
    $[0,\,1]$.
  \item[(ii)] The density function $g$ has two bounded derivatives.
  \end{itemize}
  \label{assumpG}
\item
  \begin{itemize}
  \item[(i)] There is a continuous nondecreasing function $L_1$
    satisfying $L_1(-\infty) = 0$ and $L_1(\tau_0) < \infty$ such
    that
\begin{equation*}
M(t\,|\,x) - M(s\,|\,x) \leq L_1(t) - L_1(s),
 \qquad x \in [0,\,1],~-\infty < s < t \leq \tau_0.
\end{equation*}
  \item[(ii)] The conditional (sub)distribution functions $M$ and
    $M^1$ have continuous partial derivatives, with respect to $x$,
    $\dot M$ and $\dot M^1$, respectively, that are bounded in
    $(-\infty,\,\tau_c]\times [0,\,1]$.
  \item[(iii)] There are continuous nondecreasing functions $L_2$ and
    $L_3$ with $L_2(\tau_0) < \infty$, $L_3(\tau_0) < \infty$
    and $L_2(-\infty) = L_3(-\infty) = 0$ such that
\begin{gather*}
\dot M(t\,|\,x) - \dot M(s\,|\,x) \leq L_2(t) - L_2(s),
 \qquad x \in [0,\,1],~-\infty < s < t \leq \tau_0,\\
\text{and}\\
\dot M^1(t\,|\,x) - \dot M^1(s\,|\,x) \leq L_3(t) - L_3(s),
 \qquad x \in [0,\,1],-\infty < s < t \leq \tau_0.
\end{gather*}
  \item[(iv)] The second partial derivatives, with respect to $x$, of the
    conditional (sub)distribution functions $M$ and $M^1$ exist and
    are bounded in $(-\infty,\,\tau_0]\times [0,\,1]$.
  \end{itemize}
  \label{assumpMandM1}
\item The conditional distribution functions $M$ and $M^1$ admit
  bounded Lebesgue density functions.
  \label{assumpMandM1Densities}
\item The (conditional) distribution functions $P(Z \leq t)$ and $P(Z
  \leq t\,|\,\delta = 1)$ are twice continuously differentiable and
  bounded away from zero in absolute value on any compact interval in
  the region $(-\infty,\,\tau_0]$, with the density function of $P(Z
  \leq t\,|\,\delta = 1)$ bounded away from zero at $t = \tau_0$.
  \label{assumpRho}
\item
  \begin{itemize}
    \item[(i)] The score function $J$ is bounded and nonnegative, and
      there are constants $0 < p_l < p_u \leq 1$ such that $J$ is bounded
      away from zero on $(p_l,\,p_u)$ but equal to zero on $[0,\,p_l]
      \cup [p_u,\,1]$ (when $p_u = 1$ we only require that $J$ is
      equal to zero on $[0,\,p_l]$).
    \item[(ii)] $J$ is continuously differentiable with bounded
      derivative $J'$.
  \end{itemize}
  \label{assumpJ}
\end{enumerate}

Assumptions \ref{assumpK} and \ref{assumpG} are common assumptions
taken for nonparametric regression models, which guarantee good
performance of nonparametric function estimators. Note that
Assumptions \ref{assumpMandM1} (i) and (iii) are satisfied for many
distributions. Suppose that $M$ is the logistic distribution
function with a positive, bounded mean function $m(x)$ and scale
function $s \equiv 1$. Write $l_m = \inf_x m(x)$ and $u_m = \sup_x
m(x)$. Then Assumption \ref{assumpMandM1} (i) is satisfied by choosing
$L_1(t) = \int_{-\infty}^t\,\exp(u_m - s)\{1 + \exp(l_m -
s)\}^{-2}\,ds$. When $m$ is also differentiable with a bounded
derivative, then bounding functions $L_2$ and $L_3$ (that are similar
to $L_1$) can be chosen to satisfy Assumption \ref{assumpMandM1} (iii)
as well. Assumptions \ref{assumpMandM1} (ii) and (iv) and
\ref{assumpMandM1Densities} imply the conditional distribution
functions $Q_u$ and $P(C \leq t \,|\, X)$ also meet these conditions
and that $\pi$ must meet Assumptions \ref{assumpMandM1} (ii) and (iv),
when these assumptions are required; e.g.\ due to the conditional
independence of $Y$ and $C$ given $X$ we can write
\begin{equation*}
1 - M(t\,|\,X)
 = \big[\pi(X) + \{1 - \pi(X)\}\{1 - Q_u(t\,|\,X)\}\big]P(C > t\,|\,X).
\end{equation*}
In addition, $m$ and $s$ defined in Section \ref{estimates} are
functionals based on truncated means, which implies the integrals are
restricted to compact subsets of $(-\infty,\,\tau_0]$. Therefore,
combining the Leibniz integral rule for differentiation with
Assumptions \ref{assumpMandM1} (ii) and (iv) yields that both $m$ and
$s$ are twice differentiable with bounded derivatives. Assumption
\ref{assumpRho} is a technical assumption required for the consistency
of $\maxZ1$ for $\tau_0$, and many probability distributions satisfy
this assumption as well. Finally, Assumption \ref{assumpJ} is a
standard assumption that ensures $m$ and $s$ are well-defined
$L$--type regression functionals (see page 186 of Dabrowska,
1987).

Define
\begin{equation*}
\zeta\big(x,Z_j,\delta_j,t\big)
 = \frac{\delta_j\1[Z_j \leq t]}{1 - M(Z_j-\,|\,x)}
 - \int_{-\infty}^{t}
 \frac{\1[Z_j > s]}{\{1 - M(s-\,|\,x)\}^2} M^1(ds\,|\,x),
\qquad j=1,\ldots,n.
\end{equation*}
Our first result specifies the asymptotic order and expansion of the
estimator $\hpi$, which is given in Theorem 3 of L\'opez-Cheda et
al. (2017). We offer an alternative proof of this result, which may be
found in Section \ref{appendix}.

\begin{proposition} \label{prophpi}
Let Assumptions \ref{assumpbw} -- \ref{assumpRho} hold. Then
\begin{equation*}
\sup_{x \in [0,\,1]}\big|\hpi(x) - \pi(x)\big|
 = O\big((na_n)^{-1/2}\log^{1/2}(n)\big),
 \qquad\text{a.s.}
\end{equation*}
Additionally,
\begin{equation*}
\hpi(x) - \pi(x) = -\frac{\pi(x)}{g(x)} \frac1{na_n} \sj
 K\bigg(\frac{x - X_j}{a_n}\bigg)
 \zeta\big(x,Z_j,\delta_j,\tau_0\big) + R_{1,n}(x),
\end{equation*}
where $\sup_{x \in [0,\,1]}|R_{1,n}(x)| =
O((na_n)^{-3/4}\log^{3/4}(n))$, almost surely.
\end{proposition}

In the next two results, the asymptotic orders and expansions of the
location estimator $\mhat$ and the scale estimator $\shat$ are given.

\begin{proposition} \label{propmhat}
Let Assumptions \ref{assumpbw} -- \ref{assumpRho} and Assumption
\ref{assumpJ} (i) hold. Then
\begin{equation*}
\sup_{x \in [0,\,1]} \Big| \mhat(x) - m(x) \Big|
 = O\big((na_n)^{-1/2}\log^{1/2}(n)\big),
 \quad\text{a.s.}
\end{equation*}
Additionally, if Assumption \ref{assumpJ} (ii) holds,
\begin{align*}
\mhat(x) - m(x) &= -\frac1{g(x)na_n} \sj
 K\bigg(\frac{x - X_j}{a_n}\bigg)
 \int_{-\infty}^{\tau_0}\,\zeta(x,Z_j,\delta_j,y)
 \frac{1 - Q(y\,|\,x)}{1 - \pi(x)}
 J\big(Q_u(y\,|\,x)\big)\,dy \\
&\quad + \frac{\pi(x)}{1 - \pi(x)} \frac{C_m(x)}{g(x)na_n} \sj
 K\bigg(\frac{x - X_j}{a_n}\bigg) \zeta(x,Z_j,\delta_j,\tau_0)
 + R_{2,n}(x),
\end{align*}
where $\sup_{x \in [0,\,1]} |R_{2,n}(x)| =
O((na_n)^{-3/4}\log^{3/4}(n))$, almost surely, and
\begin{equation*}
C_m(x) = \int_{-\infty}^{\tau_0}\,Q_u(y\,|\,x)
 J\big(Q_u(y\,|\,x)\big)\,dy.
\end{equation*}
\end{proposition}

\begin{proposition} \label{propshat}
Let Assumptions \ref{assumpbw} -- \ref{assumpRho} and Assumption
\ref{assumpJ} (i) hold. Then
\begin{equation*}
\sup_{x \in [0,\,1]} \Big| \shat(x) - s(x) \Big|
 = O\big((na_n)^{-1/2}\log^{1/2}(n)\big),
 \quad\text{a.s.}
\end{equation*}
Additionally, if Assumption \ref{assumpJ} (ii) holds,
\begin{align*}
\shat(x) - s(x) &= -\frac1{g(x)na_n} \sj
 K\bigg(\frac{x - X_j}{a_n}\bigg)
 \int_{-\infty}^{\tau_0}\,\zeta(x,Z_j,\delta_j,y)
 \frac{1 - Q(y\,|\,x)}{1 - \pi(x)}
 \frac{y - m(x)}{s(x)}
 J\big(Q_u(y\,|\,x)\big)\,dy \\
&\quad + \frac{\pi(x)}{1 - \pi(x)} \frac{C_s(x)}{g(x)na_n} \sj
 K\bigg(\frac{x - X_j}{a_n}\bigg) \zeta(x,Z_j,\delta_j,\tau_0)
 + R_{3,n}(x),
\end{align*}
where $\sup_{x \in [0,\,1]}
|R_{3,n}(x)| = O((na_n)^{-3/4}\log^{3/4}(n))$, almost surely, and 
\begin{equation*}
C_s(x) = \int_{-\infty}^{\tau_0}\,Q_u(y\,|\,x)
 \frac{y - m(x)}{s(x)}J\big(Q_u(y\,|\,x)\big)\,dy.
\end{equation*}
\end{proposition}

To continue, it is common in heteroskedastic models to place a
restriction on the curvature of the distribution function of either
the responses or the errors (see, for example, Chown, 2016, who works
with finite Fisher information for location and scale). Recall the
functions $L_1$, $L_2$ and $L_3$ from Assumption \ref{assumpMandM1}
(i) and (iii). We will require the function $L = L_1 + L_2 + L_3$ to
satisfy the following curvature restriction that is analogous to
assuming finite Fisher information for location and scale, i.e.\ we
assume that $L$ has two derivatives such that
\begin{equation} \label{assumpLcurve}
\int_{-\infty}^{\tau_0}
 \big(1 + v^2\big)\bigg\{\frac{L''(v)}{L'(v)}\bigg\}^2\,L(dv)
 < \infty,
 \quad \int_{-\infty}^{\tau_0} \big(1 + v^2\big)\,L(dv) < \infty
 \quad\text{and}\quad
 \sup_{-\infty < t \leq \tau_0} \big|tL'(t)\big| < \infty.
\end{equation}
Consequently, for sequences of real numbers $\{u_n\}_{n \geq 1}$ and
$\{v_n\}_{n \geq 1}$ satisfying $u_n \to u$ and $v_n \to v$, as $n \to
\infty$ and with $u,v \in \R$, \eqref{assumpLcurve} implies $|L(tv_n +
u_n) - L(tv + u)| = O(|u_n - u| + |v_n - v|)$, uniformly in $t$ (see the
proof of Theorem \ref{thmFhat} in Section \ref{appendix}). Note,
\eqref{assumpLcurve} is also satisfied for many distributions,
which includes the logistic distribution as in the example given above
of a distribution satisfying Assumptions \ref{assumpMandM1} (i) and
(iii). We are now prepared to state the two main results of this
section: a strong uniform representation of the difference $\hFemp -
F$ and the weak convergence of $n^{1/2}\{\hFemp - F\}$.

\begin{theorem}[{\sc strong uniform representation of $\hFemp - F$}]
\label{thmFhat}
Let Assumptions \ref{assumpbw} -- \ref{assumpJ} hold. Assume the
function $L = L_1 + L_2 + L_3$, where the functions $L_1$, $L_2$ and
$L_3$ are given in Assumption \ref{assumpMandM1}, is twice
differentiable and satisfies \eqref{assumpLcurve}. Finally, let $F$
satisfy \eqref{assumpLcurve}, i.e.\ $F$ has finite Fisher information
for both location and scale and the error density $f$ is bounded and
satisfies $\sup_{-\infty < t \leq \tau_F} |tf(t)| < \infty$. Then
\begin{equation*}
\hFemp(t) - F(t) = E_n(t) + R_{4,n}(t),
\end{equation*}
where $\sup_{-\infty < t \leq \tau_F} |R_{4,n}(t)| =
O((na_n)^{-3/4}\log^{3/4}(n))$, almost surely, and
\begin{equation*}
E_n(t) = T_n(t) - f(t)\big\{U_n + tV_n\big\} - W_n(t)
\end{equation*}
with
\begin{equation*}
T_n(t) = \frac1{n^2a_n} \sum_{1 \leq j,k \leq n}
 K\bigg(\frac{X_j - X_k}{a_n}\bigg)
 \frac{1 - Q(ts(X_j) + m(X_j)\,|\,X_j)}{1 - \pi(X_j)}
 \frac{\zeta(X_j,Z_k,\delta_k,ts(X_j) + m(X_j))}{g(X_j)},
\end{equation*}
\begin{equation*}
U_n = \frac1{n^2a_n} \sum_{1 \leq j,k \leq n}
 K\bigg(\frac{X_j - X_k}{a_n}\bigg)\frac1{g(X_j)}
 \int_{-\infty}^{\tau_0}\,\zeta(X_j,Z_k,\delta_k,y)
 \frac{1 - Q(y\,|\,X_j)}{1 - \pi(X_j)}
 J\big(Q_u(y\,|\,X_j)\big)\,dy,
\end{equation*}
\begin{equation*}
V_n = \frac1{n^2a_n} \sum_{1 \leq j,k \leq n}
 K\bigg(\frac{X_j - X_k}{a_n}\bigg) \frac1{g(X_j)}
 \int_{-\infty}^{\tau_0}\,\zeta(X_j,Z_k,\delta_k,y)
 \frac{1 - Q(y\,|\,X_j)}{1 - \pi(X_j)}
 \frac{y - m(X_j)}{s(X_j)}
 J\big(Q_u(y\,|\,X_j)\big)\,dy
\end{equation*}
and
\begin{equation*}
W_n(t) = \frac1{n^2a_n} \sum_{1 \leq j,k \leq n}
 K\bigg(\frac{X_j - X_k}{a_n}\bigg)
 \frac{\pi(X_j)}{1 - \pi(X_j)}
 \frac{F(t) - f(t)\{C_m(X_j) + tC_s(X_j)\}}{g(X_j)}
 \zeta(X_j,Z_k,\delta_k,\tau_0).
\end{equation*}
\end{theorem}

\begin{theorem}[{\sc weak convergence of $n^{1/2}\{\hFemp - F\}$}]
\label{thmFhatWC}
Under the conditions of Theorem \ref{thmFhat}, if the bandwidth
sequence $\{a_n\}_{n \geq 1}$ is chosen such that $a_n^2 =
o(n^{-1/2})$ and $(na_n)^{-3/4}\log^{3/4}(n) = o(n^{-1/2})$ (e.g.,
$a_n = O(n^{-1/4 - \gamma} \log^{1/4 + \gamma}(n))$ for any $0 <
\gamma < 1/12$) then $n^{1/2}\{\hFemp - F\}$ is asymptotically linear,
i.e.\
\begin{equation*}
n^{1/2}\{\hFemp(t) - F(t)\} = n^{-1/2} \sj b_t(X_j,Z_j,\delta_j)
 + R_{5,n}(t),
\end{equation*}
where $\sup_{-\infty < t \leq \tau_F}|R_{5,n}(t)| = \op$ and the
influence function is
\begin{align*}
b_t(X_j,Z_j,\delta_j) &=
 \frac{1 - Q(ts(X_j) + m(X_j)\,|\,X_j)}{1 - \pi(X_j)}
 \zeta\big(X_j,Z_j,\delta_j,ts(X_j) + m(X_j)\big) \\
&\quad - f(t) \int_{-\infty}^{\tau_0}\,
 \zeta\big(X_j,Z_j,\delta_j,y\big)
 \frac{1 - Q(y\,|\,X_j)}{1 - \pi(X_j)}
 J\big(Q_u\big(y\,\big|\,X_j\big)\big)\,dy \\
&\quad - tf(t) \int_{-\infty}^{\tau_0}\,
 \zeta\big(X_j,Z_j,\delta_j,y\big)
 \frac{1 - Q(y\,|\,X_j)}{1 - \pi(X_j)}
 \frac{y - m(X_j)}{s(X_j)}
 J\big(Q_u\big(y\,\big|\,X_j\big)\big)\,dy \\
&\quad - \frac{\pi(X_j)}{1 - \pi(X_j)}
 \Big[ F(t) - f(t)\big\{C_m(X_j) + tC_s(X_j)\big\} \Big]
 \zeta\big(X_j,Z_j,\delta_j,\tau_0\big),
\end{align*}
with $-\infty < t \leq \tau_F$. Consequently, the process
$\{n^{1/2}\{\hFemp(t) - F(t)\}\,:\,-\infty < t \leq \tau_F\}$ weakly
converges to a mean zero Gaussian process $\{Z(t)\,:\,-\infty < t \leq
\tau_F\}$ with covariance function $\Sigma(t,v) =
E[b_t(X,Z,\delta)b_v(X,Z,\delta)]$ for $-\infty < t,v \leq \tau_F$.
\end{theorem}

\begin{remark}[{\sc consequences for the choice of bandwidth}]
\label{remBandwidthChoice}
Theorem \ref{thmFhatWC} implies that the estimator
$\hFemp$ is a root-$n$ consistent estimator of $F$ only when the
bandwidth sequence $\{a_n\}_{n \geq 1}$ satisfies $a_n^2 =
o(n^{-1/2})$ and $(na_n)^{-3/4}\log^{3/4}(n) = o(n^{-1/2})$, which
{\em undersmoothes} the estimators $\hQ$, $\mhat$, $\shat$ and
$\hpi$. A bandwidth sequence given by $a_n = O(n^{-1/4 -
  \gamma}\log^{1/4 + \gamma}(n))$ satisfies $a_n^2 = o(n^{-1/2})$ and
$(na_n)^{-3/4}\log^{3/4}(n) = o(n^{-1/2})$ for every $0 < \gamma <
1/12$. Note that when $\gamma = 1/12$ we have $a_n =
O(n^{-1/3}\log^{1/3}(n))$, and this choice does not lead to a root-$n$
consistent estimator because this bandwidth undersmoothes by too
much. Alternatively, when $\gamma = 0$ we have $a_n =
O(n^{-1/4}\log^{1/4}(n))$, and this choice also does not lead to a
root-$n$ consistent estimator because this bandwidth does not
undersmooth by enough. Another interesting consequence highlighted by
Theorem \ref{thmFhatWC} is that the asymptotic behavior of
$n^{1/2}\{\hFemp - F\}$ does not depend on the bandwidth sequence
$\{a_n\}_{n \geq 1}$ used to construct the covariate-localized
estimators when $a_n^2 = o(n^{-1/2})$ and $(na_n)^{-3/4}\log^{3/4}(n)
= o(n^{-1/2})$.
\end{remark}


\section{Applications of the previous results}
\label{applications}

Here we investigate the finite sample properties of the proposed
estimator $\hFemp$ using a simulation study. Our results indicate good
finite sample behavior even at the smaller sample size of $100$ with
multiple bandwidth configurations. This is particularly encouraging
as we did not need to perform a computationally costly bandwidth
selection procedure. Instead, we consider a variety of bandwidths of
the form $C\hat{\sigma}_{X} n^{-1/4 - \gamma}\log^{1/4 + \gamma}(n)$,
with parameters $C$ and $\gamma$ arbitrarily chosen and
$\hat{\sigma}_{X}$ denoting the sample standard deviation of the
covariates. The results given here reflect the conclusion in Remark
\ref{remBandwidthChoice}: the estimator $\hFemp$ shows insensitivity to
the choice of bandwidth parameters $\{a_n\}_{n \geq 1}$ used to
construct the local model estimators when this sequence is chosen to
appropriately undersmooth these estimators. This section is concluded
with an illustration of the previous results using a dataset collected
to study the behavior of distant metastasis in lymph-node-negative
breast cancer sufferers.

\subsection{Numerical study}
\label{simulations}
To study the finite sample performance of the estimator $\hFemp$, we
conducted simulations of $1000$ runs using sample sizes $100$, $200$,
$500$ and $1000$ under the following data generation scheme. The
covariates $X$ are uniformly distributed on the interval $[-1,\,1]$,
and the location and scale functions are chosen as
\begin{equation} \label{sims_ms}
m(x) = 1 + 2x + \frac54 \cos\big(\pi x^2\big)
\quad\text{and}\quad
s(x) = 1 + \frac12 \cos\big(\pi x\big),
\qquad x \in [-1,\,1].
\end{equation}
For the error distribution, we chose the standard normal
distribution that has been truncated at 2, centered at
zero and scaled to satisfy the cure model identifiability
requirements (see Section \ref{estimates}). An initial set of
responses $Y$ are then obtained using \eqref{modeleq}.

We work with a cure proportion function given by the logistic
distribution function with standard scaling that has been centered at
$7 / 4$, which gives about $16\%$ cured cases on average. Cure
indicators are randomly generated based on this probability function,
and whenever a cure indicator is equal to one we replace the
corresponding value of $Y$ with $\infty$. Finally, the censoring
variables $C$ are randomly generated from a mixture distribution of two
components with equal mixing probabilities, where one component
distribution is a normal distribution centered at $10$ with variation
$1 / 2$ and the other component distribution is a shifted version of
\eqref{modeleq}, with $m$ and $s$ as in \eqref{sims_ms} but now $m$ is
shifted up by $1/2$ and the model errors are standard normally
distributed (no truncation). These choices result in about $18\%$
censored values for the uncured cases. When we combine the censoring
from both cases (cured and uncured) we expect a typical dataset
generated in our simulations to present with about $31\%$ of censored
response values.

The resulting response values $Z$ are taken as the minimum of each $Y$
and $C$ and a censoring indicator $\delta$ is set equal to one
whenever $Y \leq C$ and zero otherwise. Finally, the score function
$J$ is chosen by regions of $[0,\,1]$. In the region $(0.0001,\,1]$,
$J$ is chosen as the logistic distribution function with scaling
$0.0001$, and, in the region $[0,\,0.0001]$, $J$ is set equal to
$0$. This is a smooth approximation of a step function that nearly
integrates to $1$.

The bandwidth parameter sequence $\{a_n\}_{n \geq 1}$ used to
construct the covariate-localized model estimators is of the form $a_n
= C\hat{\sigma}_{X} n^{-1/4 - \gamma} \log^{1/4 + \gamma}(n)$, where
$\hat{\sigma}_{X}$ denotes the sample standard deviation of the
covariates $X$. We investigate four situations: the constant of
proportionality $C$ is either $3/4$ or $9/8$ and the exponent
parameter $\gamma$ is either $1/16$ or $1/28$. These choices are
appropriate for Theorem \ref{thmFhatWC}, since $1/28 < 1/16 < 1/12$.

We numerically measure the performance of the estimator $\hFemp$ in two
ways. First, the asymptotic mean squared errors at $t$-values $-2$,
$-1$, $0$, $1$ and $2$ are simulated, where this performance metric is
calculated by first obtaining the simulated mean squared errors and
then multiplying these by the corresponding sample size. Second, the
asymptotic integrated mean squared error is simulated, where this
quantity is calculated similarly to the asymptotic mean squared errors
but now we integrate over $t$. These performance metrics are predicted
to be stable from Theorem \ref{thmFhatWC}.

\begin{table}
\centering
\begin{tabular}{|r|c|c c c c c|}
\hline
$n$    & $C,\,\gamma$ & $-2$ & $-1$ & $0$ & $1$ & $2$ \\
\hline
$100$  & $3/4,\,1/16$ & $0.014$ & $0.046$ & $0.078$ & $0.054$ & $0.003$ \\
       & $3/4,\,1/28$ & $0.014$ & $0.041$ & $0.073$ & $0.051$ & $0.004$ \\
       & $9/8,\,1/16$ & $0.015$ & $0.030$ & $0.082$ & $0.038$ & $0.006$ \\
       & $9/8,\,1/28$ & $0.014$ & $0.028$ & $0.079$ & $0.032$ & $0.006$ \\
\hline
$200$  & $3/4,\,1/16$ & $0.014$ & $0.040$ & $0.078$ & $0.043$ & $0.005$ \\
       & $3/4,\,1/28$ & $0.014$ & $0.036$ & $0.070$ & $0.037$ & $0.006$ \\
       & $9/8,\,1/16$ & $0.016$ & $0.030$ & $0.086$ & $0.030$ & $0.011$ \\
       & $9/8,\,1/28$ & $0.017$ & $0.026$ & $0.095$ & $0.032$ & $0.009$ \\
\hline
$500$  & $3/4,\,1/16$ & $0.015$ & $0.035$ & $0.088$ & $0.040$ & $0.009$ \\
       & $3/4,\,1/28$ & $0.016$ & $0.036$ & $0.088$ & $0.038$ & $0.010$ \\
       & $9/8,\,1/16$ & $0.019$ & $0.032$ & $0.099$ & $0.031$ & $0.018$ \\
       & $9/8,\,1/28$ & $0.020$ & $0.031$ & $0.111$ & $0.029$ & $0.023$ \\
\hline
$1000$ & $3/4,\,1/16$ & $0.017$ & $0.043$ & $0.093$ & $0.047$ & $0.012$ \\
       & $3/4,\,1/28$ & $0.017$ & $0.039$ & $0.101$ & $0.041$ & $0.014$ \\
       & $9/8,\,1/16$ & $0.019$ & $0.034$ & $0.109$ & $0.032$ & $0.026$ \\
       & $9/8,\,1/28$ & $0.023$ & $0.032$ & $0.126$ & $0.033$ & $0.033$ \\
\hline
\end{tabular}
\caption{Simulated asymptotic mean squared error values of $\hFemp$ at
  the points $-2$, $-1$, $0$, $1$ and $2$.}
\label{table_Fhat_AMSE}
\end{table}

\begin{table}
\centering
\begin{tabular}{|r|c c c c|}
\hline
\diagbox{$C,\,\gamma$}{$n$} & $100$  & $200$  & $500$  & $1000$ \\
\hline
$3/4,\,1/16$ & $0.209$ & $0.189$ & $0.199$ & $0.213$ \\
$3/4,\,1/28$ & $0.195$ & $0.178$ & $0.192$ & $0.206$ \\
$9/8,\,1/16$ & $0.175$ & $0.176$ & $0.197$ & $0.214$ \\
$9/8,\,1/28$ & $0.161$ & $0.179$ & $0.203$ & $0.236$ \\
\hline
\end{tabular}
\caption{Simulated asymptotic integrated mean squared error values of
  $\hFemp$.}
\label{table_Fhat_AMISE}
\end{table}

The results of the simulated asymptotic mean squared errors are
displayed in Table \ref{table_Fhat_AMSE}, and the results of the
simulated asymptotic mean integrated squared errors are given in
Table \ref{table_Fhat_AMISE}. The values in Table
\ref{table_Fhat_AMSE} show the estimator $\hFemp$ has asymptotically
stable pointwise mean squared errors (at the $t$-values $-2$, $-1$, $0$,
$1$ and $2$), and this metric clearly shows the estimator $\hFemp$ to
have good performance on samples as small as $100$. The values in
Table \ref{table_Fhat_AMISE} show a strong mirroring of the
conclusions drawn from Table \ref{table_Fhat_AMSE}, which indicate
that $\hFemp$ is a good estimator of $F$ even for samples as small as
$100$.

We tried other bandwidth configurations and found similar
results. However, in some cases, the performance metrics above were
affected. Specifically, choosing the constant of proportionality $C$
either too large or too small showed the most significant effects
while changing the exponent $\gamma$ showed no practical
effect. We observed that choosing $C$ too large negatively impacted
large sample behavior ($n = 1000$) and choosing $C$ too small
negatively impacted small sample behavior ($n = 100$). This effect can
be seen in Table \ref{table_Fhat_AMISE} for the rows corresponding to
$C = 9 / 8$, but it is not very pronounced in this case. This indeed
reflects the conclusions of Remark \ref{remBandwidthChoice} that state
the bandwidth should undersmooth but not undersmooth by too
much. Nevertheless, the estimator $\hFemp$ does show insensitivity to
the choice of bandwidth when this parameter is appropriately
chosen. We therefore expect that a simple rule-of-thumb approach can
be an effective strategy for choosing an appropriate bandwidth, where
one (say) compares plots of several estimators of $F$ and chooses a
bandwidth parameter among those that produced very similar estimators
of $F$.

Summing up, we have numerically demonstrated that the estimator
$\hFemp$ has good finite sample performance with samples sizes as
small as $100$. Our numerical results show the bandwidth parameter
sequence $\{a_n\}_{n \geq 1}$ used to construct the
covariate-localized model estimators does not appear to have strong
influence on the behavior of $\hFemp$ even at smaller samples. A
possible explanation for this behavior is that we are averaging over
the local estimators of $F$; see \eqref{Fhat} for the definition of
$\hFemp$. The estimator $\hFemp$ shows strong potential for use in
applications where the unknown error distribution function $F$
requires estimation; e.g.\ testing model assumptions, building
confidence intervals / bands, etc.


\subsection{Analysis of breast cancer data}
In this section we illustrate the estimators of the components from
model \eqref{modeleq}, i.e.\ $\pi$, $m$, $s$ and $F$, using a set of
data obtained from frozen samples of tumour tissue stored at the
Erasmus Medical Center (Rotterdam, Netherlands) of patients who were
treated for lymph-node negative breast cancer during $1980$--$95$.
It has been observed that about $60\%$--$70\%$ of patients treated
are cured (see page $671$ of Wang et al., 2005). These data were
collected to study distant metastasis of lymph-node-negative breast
cancer sufferers, where it is desirable to (for example) identify
medical treatments that increase the amount of time before (possible)
metastasis occurs. See Wang et al.\ (2005) for a complete description
of these data.

\begin{figure}
\includegraphics[width=0.60\textwidth]{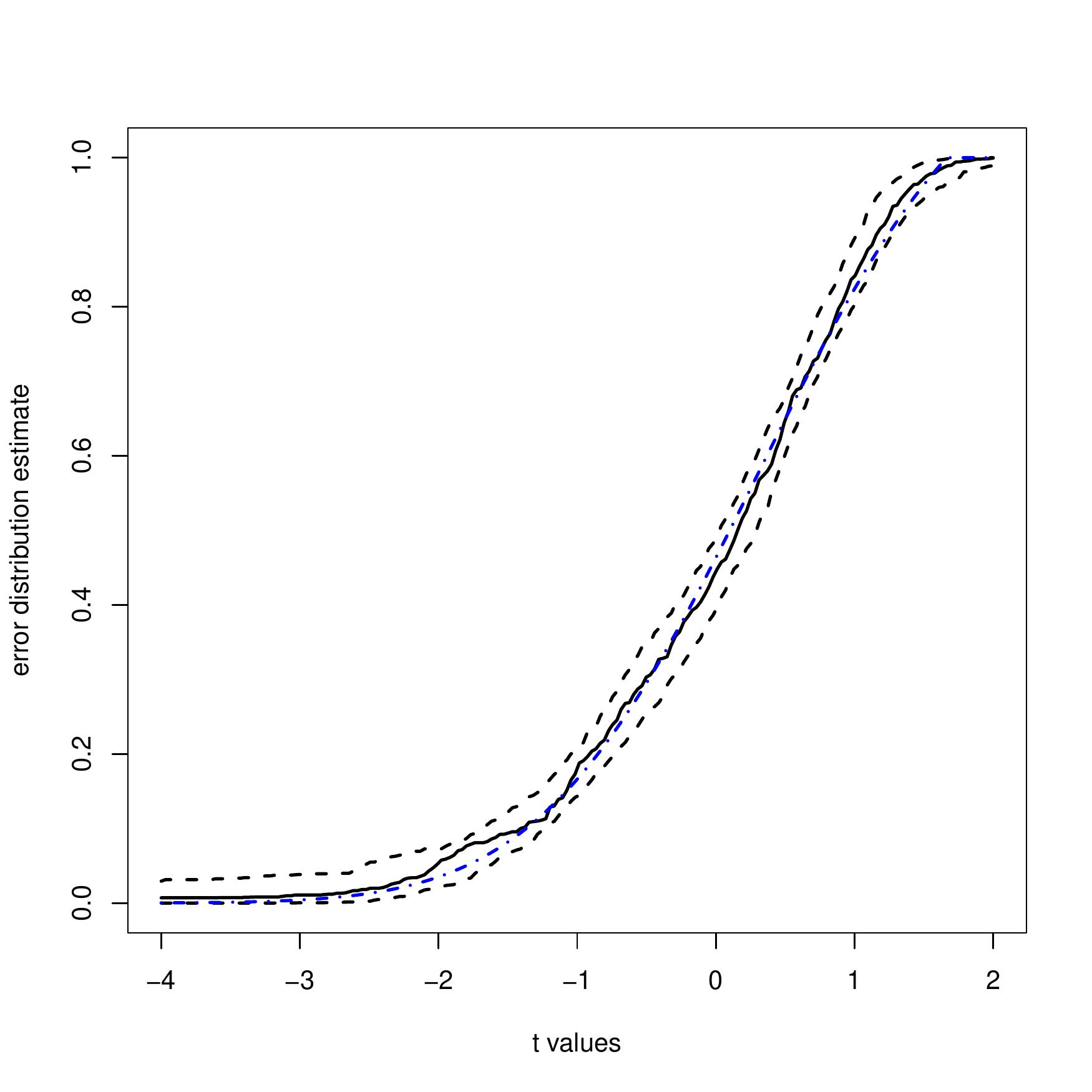}
\vspace{1ex}
\caption{A plot of the error distribution function estimate (solid
  black curve) overlaid by approximate $95\%$ pointwise
  confidence intervals (dashed curves) and a plausible truncated normal
  error distribution (blue dot-dashed curve).}
\label{veridex_analysis}
\end{figure}

Our analysis considers three variables measured: the number of days
before metastasis was detected (observed or censored), a censoring
indicator and the patient's age (in years). There are $286$ original
data values, and about $63\%$ of the reported time lengths before
metastasis had been detected were censored at large values (indicating
a possible cure effect). The ages of the patients range between $26$
and $83$ years with a median age of $52$ years. The oldest patient
with an uncensored response is $78$ years. Since there were only
two (censored) observations for patients older than $80$, these cases
were removed because the data was too sparse in this region to obtain
good model estimates, and our analysis considers the remaining $284$
patients. We are interested in a nonparametric location-scale
modelling of the log-transformed time length before detectable
metastasis by the patient's age that accounts for both the presence of
censoring and an apparent presence of a cure effect.

We obtain from $a_{n} = C\hat{\sigma}_{X}n^{-1/4 - \gamma} \log^{1/4 +
\gamma}(n)$, with $n = 284$ and $\hat{\sigma}_{X} \approx 12.3$
years and choosing $C = 9/8$ and $\gamma = 1/28$, a bandwidth of
$4.51$ years. This choice corresponds with our simulations
from the previous section and corresponds with an appropriate choice
for Theorem \ref{thmFhatWC}. The score function $J$ is chosen as in
the previous section (a smooth approximation of a step function).

Pointwise confidence intervals for $\hFemp$ are built using a
bootstrap as follows. We begin by sampling completely at random and
with replacement from the ages of the patients (covariates). We then
construct bootstrap uncured responses using model \eqref{modeleq},
where $m$ is replaced by the estimator $\mhat$, $s$ is replaced by the
estimator $\shat$ and the model errors are sampled independently from
$\hFemp$ and then appropriately centered and scaled (see our
discussion on identifiability in Section \ref{estimates}). In
addition, a bootstrap cure indicator is independently generated from
$\hpi$ for each selected age. When this indicator is equal to one the
associated bootstrap uncured response value is replaced by
$\infty$. Bootstrap censoring variables are then independently sampled
from the local censoring distribution estimators at each selected
age. The bootstrap responses are then obtained by taking the minimum
between the resulting augmented uncured responses (that include
$\infty$ as possible values) and the selected censoring
variable. Whenever a response is not censored we set a bootstrap
censoring indicator equal to one and zero otherwise. The bootstrap
distribution of $\hFemp$ is simulated on $300$ replicate bootstrap
datasets, and we obtain our confidence intervals using the approximate
quantiles from the simulated bootstrap distribution.

The confidence intervals considered in this analysis have
approximately $95\%$ coverage.

Figure \ref{veridex_analysis} shows a plot of the error
distribution function estimate, which appears to be truncated near
$1.7$ and (comparing the solid black curve with the blue dot-dashed
curve) this estimate appears to describe very well a (truncated)
normal distribution. In conclusion, the estimator
$\hFemp$ appears to be very promising for use in applications, and
this estimator does not require a computationally costly bandwidth
selection procedure to be effective.


\section{Appendix}
\label{appendix}

From the discussion in Section \ref{estimates} we observed that
\eqref{curemodid} implies $\tau_0 < \tau_M(x)$, $x \in [0,\,1]$. This
means that we may view our cure model as a special case of
right--censored response models. In the following result, we specify
the asymptotic order of the estimators $\hM$ and $\hM^1$, which follows
directly from Lemma 4.2 of Du and Akritas (2002).

\begin{lemma}
\label{lemhMhM1Consistency}
Let Assumptions \ref{assumpbw} -- \ref{assumpG}, \ref{assumpMandM1}
(ii) and (iv) hold. Then
\begin{equation*}
\sup_{x \in [0,\,1]} \sup_{-\infty < t \leq \tau_0}
 \big|\hM(t\,|\,x) - M(t\,|\,x)\big| 
 = O\big((na_n)^{-1/2}\log^{1/2}(n)\big),
\qquad\text{a.s.},
\end{equation*}
and
\begin{equation*}
\sup_{x \in [0,\,1]} \sup_{-\infty \leq t \leq \tau_0}
 \big|\hM^1(t\,|\,x) - M^1(t\,|\,x)\big|
 = O\big((na_n)^{-1/2}\log^{1/2}(n)\big),
\qquad\text{a.s.}
\end{equation*}
\end{lemma}

In addition, we can specify the asymptotic order for moduli of
continuity for the estimators $\hM$ and $\hM^1$, which follows
directly by applications of Lemma 4.4 of Du and Akritas (2002).
\begin{lemma}
\label{lemhMhM1modulus}
Let Assumptions \ref{assumpbw} -- \ref{assumpMandM1} hold. Set $L =
L_1 + L_2$ and $b_n = O((na_n)^{-1/2}\log^{1/2}(n))$. Then, almost
surely,
\begin{equation*}
\sup_{x \in [0,\,1]} \sup_{-\infty < s, t \leq \tau_0}
 \sup_{|L(t) - L(s)| \leq b_n} \big|
 \hM(t\,|\,x) - M(t\,|\,x) - \hM(s\,|\,x) + M(s\,|\,x)\big| 
 = O\big((na_n)^{-3/4}\log^{3/4}(n)\big)
\end{equation*}
and, now with $L = L_1 + L_3$,
\begin{equation*}
\sup_{x \in [0,\,1]} \sup_{-\infty < s, t \leq \tau_0}
 \sup_{|L(t) - L(s)| \leq b_n} \big|
 \hM^1(t\,|\,x) - M^1(t\,|\,x) - \hM^1(s\,|\,x) + M^1(s\,|\,x)\big| 
 = O\big((na_n)^{-3/4}\log^{3/4}(n)\big).
\end{equation*}
\end{lemma}

In order to specify the asymptotic order and modulus of continuity for
the hazard estimator $\hHaz$, we need to state a technical result
common for censored response models. Define
\begin{equation*}
A_n(t\,|\,x) = \int_{-\infty}^{t}
 \frac{\hM(s-\,|\,x) - M(s-\,|\,x)}{\{1 - M(s-\,|\,x)\}^2}
 \big\{\hM^1(ds\,|\,x) - M^1(ds\,|\,x)\big\}.
\end{equation*}
In the following result, we specify the asymptotic order of $A_n$. The
proof of this result follows along the same lines as the proof of
Proposition 4.1 of Du and Akritas (2002), and it is therefore omitted.
\begin{proposition}
\label{propAnOrder}
Let Assumptions \ref{assumpbw} -- \ref{assumpMandM1} hold. Then
\begin{equation*}
\sup_{x \in [0,\,1]} \sup_{-\infty < t \leq \tau_0} \big|A_n(t\,|\,x)\big|
 = O\big((na_n)^{-3/4}\log^{3/4}(n)\big),
 \qquad\text{a.s.}
\end{equation*}
\end{proposition}

With the results above we can state the asymptotic order and modulus of
continuity for the hazard estimator $\hHaz$.
\begin{lemma}
\label{lemhHazConsistencyModulus}
Let Assumptions \ref{assumpbw} -- \ref{assumpMandM1} hold. Then
\begin{equation*}
\sup_{x \in [0,\,1]} \sup_{-\infty < t \leq \tau_0}
 \big|\hHaz(t\,|\,x) - \Haz(t\,|\,x)\big|
 = O\big((na_n)^{-1/2}\log^{1/2}(n)\big),
 \qquad\text{a.s.}
\end{equation*}
Set $L = L_1 + L_2 + L_3$ and $b_n =
O((na_n)^{-1/2}\log^{1/2}(n))$. Then
\begin{equation*}
\sup_{x \in [0,\,1]} \sup_{-\infty < s,t \leq \tau_0}
 \sup_{|L(t) - L(s)| \leq b_n}
 \big| \hHaz(t\,|\,x) - \Haz(t\,|\,x)
 - \hHaz(s\,|\,x) + \Haz(s\,|\,x) \big|
 = O\big((na_n)^{-3/4}\log^{3/4}(n)\big),
 \qquad\text{a.s.}
\end{equation*}
\end{lemma}
\begin{proof}
Beginning with the first assertion, we can write $\hHaz(t\,|\,x) -
\Haz(t\,|\,x) = R_1(t\,|\,x) + R_2(t\,|\,x) + R_3(t\,|\,x)$, where
\begin{equation*}
R_1(t\,|\,x) = \int_{-\infty}^t
 \frac{\hM(s-\,|\,x) - M(s-\,|\,x)}{
 \{1 - \hM(s-\,|\,x)\}\{1 - M(s-\,|\,x)\}} M^1(ds\,|\,x),
\end{equation*}
\begin{equation*}
R_2(t\,|\,x) = \int_{-\infty}^t
 \frac{\hM^1(ds\,|\,x) - M^1(ds\,|\,x)}{1 - M(s-\,|\,x)}
\end{equation*}
and
\begin{equation*}
R_3(t\,|\,x) = \int_{-\infty}^t 
 \bigg\{ \frac{\hM(s-\,|\,x) - M(s-\,|\,x)}{
 \{1 - \hM(s-\,|\,x)\}\{1 - M(s-\,|\,x)\}} \bigg\}
 \big\{\hM^1(ds\,|\,x) - M^1(ds\,|\,x)\big\}.
\end{equation*}
Condition \eqref{curemodid} implies that $M(t\,|\,x) < 1$ for all $t
\in (-\infty,\,\tau_0]$ and all $x \in [0,\,1]$. In addition, the
assumptions of Lemma \ref{lemhMhM1Consistency} are satisfied, which
implies that $\sup_{x \in [0,\,1]} \sup_{-\infty < t \leq \tau_0}
|\hM(t\,|\,x) - M(t\,|\,x)| = O((na_n)^{-1/2}\log^{1/2}(n))$, almost
surely. Combining these results, it is easy to see that $\sup_{x \in
  [0,\,1]} \sup_{-\infty < t \leq \tau_0} |R_1(t\,|\,x)| =
O((na_n)^{-1/2}\log^{1/2}(n))$, almost surely. Using integration by parts,
we can write $R_2(t\,|\,x)$ as
\begin{equation*}
\frac{\hM^1(t\,|\,x) - M^1(t\,|\,x)}{1 - M(t-\,|\,x)}
 - \int_{-\infty}^t 
 \frac{\hM^1(s\,|\,x) - M^1(s\,|\,x)}{\{1 - M(s-\,|\,x)\}^2}
 M(ds\,|\,x).
\end{equation*}
Similar to the first term, both of the terms in the display above have
the order $O((na_n)^{-1/2}\log^{1/2}(n))$, almost surely. The
assumptions of Proposition \ref{propAnOrder} are satisfied. Combining
the statement of this result with the uniform, strong consistency of
$\hM$ for $M$, it follows that $\sup_{x \in [0,\,1]} \sup_{-\infty < t
\leq \tau_0} |R_3(t\,|\,x)| = O((na_n)^{-3/4}\log^{3/4}(n)) =
o((na_n)^{-1/2}\log^{1/2}(n))$, almost surely. This concludes the
proof of the first assertion. The second assertion follows by similar
arguments and is therefore omitted.
\end{proof}

With the results above on the cumulative hazard estimator $\hHaz$, we
can state the asymptotic orders of the strong, uniform consistency and
modulus of continuity of the Beran estimator $\hQ$. In addition, the
uniform, strong i.i.d.\ representation of the Beran estimator $\hQ$ is
given.

\begin{lemma}[{\sc properties of the Beran estimator $\hQ$}]
\label{lemhQConsistencyModulusExpan}
Let Assumptions \ref{assumpbw} -- \ref{assumpMandM1Densities}
hold. Then
\begin{equation*}
\sup_{x \in [0,\,1]} \sup_{-\infty < t \leq \tau_0}
 \big|\hQ(t\,|\,x) - Q(t\,|\,x)\big|
 = O\big((na_n)^{-1/2}\log^{1/2}(n)\big),
 \qquad\text{a.s.}
\end{equation*}
Set $L = L_1 + L_2 + L_3$ and $b_n =
O((na_n)^{-1/2}\log^{1/2}(n))$. Then
\begin{equation*}
\sup_{x \in [0,\,1]} \sup_{-\infty < s,t \leq \tau_0}
 \sup_{|L(t) - L(s)| \leq b_n}
 \big| \hQ(t\,|\,x) - Q(t\,|\,x)
 - \hQ(s\,|\,x) + Q(s\,|\,x) \big|
 = O\big((na_n)^{-3/4}\log^{3/4}(n)\big),
 \qquad\text{a.s.}
\end{equation*}
Finally, with $\zeta$ defined in Section \ref{asymptotics},
\begin{equation*}
\hQ(t\,|\,x) - Q(t\,|\,x) = \frac{1 - Q(t\,|\,x)}{g(x)na_n} \sj
 K\bigg(\frac{x - X_j}{a_n}\bigg) \zeta(x,Z_j,\delta_j,t) + R_n(t\,|\,x),
\end{equation*}
where $\sup_{x \in [0,\,1]} \sup_{t \in (-\infty,\,\tau_0]}
|R_n(t\,|\,x)| = O((na_n)^{-3/4}\log^{3/4}(n))$, almost surely.
\end{lemma}
\begin{proof}
Write, as in the proof of Theorem 3.2 of Du and Akritas (2002),
\begin{align*}
\hQ(t\,|\,x) - Q(t\,|\,x) &=
 \{1 - Q(t\,|\,x)\} \int_{-\infty}^t
 \frac{1 - Q(t-\,|\,x)}{1 - Q(t\,|\,x)} \,
 d\big(\hHaz(t\,|\,x) - \Haz(t\,|\,x)\big) \\
&\quad + \{1 - Q(t\,|\,x)\} \int_{-\infty}^t
 \frac{Q(t-\,|\,x) - \hQ(t-\,|\,x)}{1 - Q(t\,|\,x)} \,
 d\big(\hHaz(t\,|\,x) - \Haz(t\,|\,x)\big).
\end{align*}
Since $Y_u$ is a continuous random variable, it follows that $1 -
Q(t-\,|\,x) = 1 - Q(t\,|\,x)$ and the first term on the right-hand side
in the display above is equal to $\{1 - Q(t\,|\,x)\}\{\hHaz(t\,|\,x) -
\Haz(t\,|\,x)\}$, which, by the first statement of Lemma
\ref{lemhHazConsistencyModulus}, has the order
$O((na_n)^{-1/2}\log^{1/2}(n))$, almost surely, uniformly over $(t,x)
\in (-\infty,\,\tau_0] \times [0,\,1]$. Following the arguments in the
proof of Theorem 3.2 of Du and Akritas (2002), the second term has the
order $O((na_n)^{-3/4}\log^{3/4}(n))$, almost surely, uniformly over
$(t,x) \in (-\infty,\,\tau_0] \times [0,\,1]$ (see the decomposition
of $(\text{B})/S_x(t)$ into expressions $(5.5)$--$(5.8)$ in that
article, where these expressions are shown to have the desired
order). This completes the proof of the first assertion. The second
assertion follows by a similar argument and is therefore
omitted. Finally, the last assertion follows directly by an
application of Theorem 3.2 of Du and Akritas (2002).
\end{proof}

L\'opez-Cheda et al. (2017) give the asymptotic order of strong
consistency of $\maxZ1$ for $\tau_0$, which we restate here for
convenience.

\begin{lemma}[{\sc Lemma 5 of L\'opez-Cheda et al. (2017)}]
\label{lemMaxZ1Consistency}
Let Assumption \ref{assumpRho} hold. Then
\begin{equation*}
n^\alpha \big(\tau_0 - \maxZ1\big) = o(1),
 \quad\text{a.s.},
 \quad \alpha \in (0,\,1).
\end{equation*}
Note, for a sequence of bandwidths $a_n$ satisfying $a_n \to 0$, as $n
\to \infty$, such that $na_n^5\log^{-1}(n) = O(1)$, it follows that
\begin{equation*}
\tau_0 - \maxZ1 = o\big((na_n)^{-3/4}\log^{3/4}(n)\big),
 \qquad\text{a.s.}
\end{equation*}
\end{lemma}

We can now state the proof of Proposition \ref{prophpi} from Section
\ref{asymptotics}.

\begin{proof}[{\sc Proof of Proposition \ref{prophpi}}]
Beginning with the first assertion, we can write $\hpi(x) - \pi(x) =
R_1(x) + R_2(x) + R_3(x)$, where
\begin{equation*}
R_1(x) = Q(\tau_0\,|\,x) - \hQ(\tau_0\,|\,x),
\end{equation*}
\begin{equation*}
R_2(x) = \hQ(\tau_0\,|\,x) - Q(\tau_0\,|\,x)
 - \hQ(\maxZ1\,|\,x) + Q(\maxZ1\,|\,x)
\end{equation*}
and
\begin{equation*}
R_3(x) = Q(\tau_0\,|\,x) - Q(\maxZ1\,|\,x)
\end{equation*}
The assumptions of Lemma \ref{lemhQConsistencyModulusExpan} are
satisfied, and the first statement of this result implies that
$\sup_{x \in [0,\,1]} |R_1(x)| = O((na_n)^{-1/2}\log^{1/2}(n))$,
almost surely. Since the assumptions of Lemma
\ref{lemMaxZ1Consistency} are satisfied, we obtain the desired
$\sup_{x \in [0,\,1]} |R_2(x)| = O((na_n)^{-3/4}\log^{3/4}(n))$,
almost surely, by combining the statement of Lemma
\ref{lemMaxZ1Consistency} with the second statement of Lemma
\ref{lemhQConsistencyModulusExpan}. Finally, combining the result of
Lemma \ref{lemMaxZ1Consistency} with the fact that $Q$ has a bounded
density shows that $\sup_{x \in [0,\,1]} |R_3(x)| =
o((na_n)^{-3/4}\log^{3/4}(n))$, almost surely. This finishes the proof
of the first assertion. The second assertion follows from applying the
third statement of Lemma \ref{lemhQConsistencyModulusExpan} to
$R_1(x)$.
\end{proof}

Next we give sketches of the proofs of Proposition \ref{propmhat} and
Proposition \ref{propshat} from Section \ref{asymptotics} that follow
along the same lines of arguments given in the proofs of Proposition
3, Proposition 6 and Proposition 7 of Akritas and Van Keilegom (2001).

\begin{proof}[{\sc Proof of Proposition \ref{propmhat}}]
We begin with proving the first assertion that $\|\mhat -
m\|_{\infty}$ has the order $O((na_n)^{-1/2}\log^{1/2}(n))$, almost
surely, writing $\|\cdot\|_{\infty}$ for the supremum norm. Following
the procedure in the proof of Proposition 3 from Akritas and Van
Keilegom (2001), write $I(q) = \int_0^q J(p)\,dp$. We have that
$\mhat(x) - m(x)$ is equal to
\begin{align} \nonumber
&\int_0^1 \hxi\big(\big\{1 - \hpi(x)\big\}p\,\big|\,x\big)J(p)\,dp
 - \int_0^1 \xi(\{1 - \pi(x)\}p\,|\,x)J(p)\,dp \\ \nonumber
&= \int_0^1 \int_{-\infty}^{\tau_0}
 \1\bigg[\frac{\hQ(t\,|\,x)}{1 - \hpi(x)} \leq p\bigg]J(p)\,dt\,dp
 - \int_0^1 \int_{-\infty}^{\tau_0}
 \1\bigg[\frac{Q(t\,|\,x)}{1 - \pi(x)} \leq p\bigg]J(p)\,dt\,dp \\
\label{mhatmIntegralForm}
&= \int_{-\infty}^{\tau_0} \bigg\{
 I\bigg(\frac{Q(t\,|\,x)}{1 - \pi(x)}\bigg)
 - I\bigg(\frac{\hQ(t\,|\,x)}{1 - \hpi(x)}\bigg) \bigg\}\,dt,
\end{align}
where we have used that $I(1) = \int_0^1 J(p)\,dp = 1$ in the second
equality. Since the score function $J$ is bounded, it follows that $I$
is Lipschitz continuous (with constant $\|J\|_{\infty}$)
and $\|\mhat - m\|_{\infty}$ is bounded by
\begin{equation*}
\bigg[ \inf_{x \in [0,\,1]} \inf_{p_l \leq p \leq p_u}
 q\big(\xi(\{1 - \pi(x)\}p\,|\,x)\,|\,x\big) \bigg]^{-1}
 \sup_{p_l \leq p \leq p_u} J(p)
 \sup_{x \in [0,\,1]} \sup_{-\infty < t \leq \tau_0}
 \bigg| \frac{\hQ(t\,|\,x)}{1 - \hpi(x)}
 - \frac{Q(t\,|\,x)}{1 - \pi(x)} \bigg|,
\end{equation*}
where the first two terms are finite. The last term is easily shown to
be of order $O((na_n)^{-1/2}\log^{1/2}(n))$, almost surely, using the
first statement of Lemma \ref{lemhQConsistencyModulusExpan} and the
first statement of Proposition \ref{prophpi}. This completes the proof
of the first assertion.

Turning now to the second assertion, we can use the procedure in the
proof of Proposition 6 of Akritas and Van Keilegom (2001); i.e.\ we
can write $\mhat(x) - m(x)$ as
\begin{equation} \label{mhatdiffmremainder}
\int_{-\infty}^{\tau_0} \bigg\{
 \frac{Q(t\,|\,x)}{1 - \pi(x)} - \frac{\hQ(t\,|\,x)}{1 - \hpi(x)}
 \bigg\}J\big(Q_u(t\,|\,x)\big)\,dt
 + O\big((na_n)^{-1}\log(n)\big),
 \qquad\text{a.s.,}
\end{equation}
using a Taylor expansion of $I$ in the right-hand side of
\eqref{mhatmIntegralForm} and the fact that $J'$ is bounded. The
difference term from \eqref{mhatdiffmremainder} becomes
\begin{equation} \label{mhatExpanDiffTerms}
\frac{Q(t\,|\,x) - \hQ(t\,|\,x)}{1 - \pi(x)}
 - Q_u(t\,|\,x)\frac{\hpi(x) - \pi(x)}{1 - \pi(x)}
 + O\big((na_n)^{-1}\log(n)\big),
 \qquad\text{a.s.,}
\end{equation}
where we have used that $Q(t\,|\,x) = \{1 - \pi(x)\}Q_u(t\,|\,x)$. The
third statement of Lemma \ref{lemhQConsistencyModulusExpan} implies
the first term in \eqref{mhatExpanDiffTerms} is equal to
\begin{equation} \label{mhatExpanResult1}
-\frac{1 - Q(t\,|\,x)}{1 - \pi(x)} \frac1{g(x)na_n} \sj
 K\bigg(\frac{x - X_j}{a_n}\bigg)\zeta(x,Z_j,\delta_j,t)
 + O\big((na_n)^{-3/4}\log^{3/4}(n)\big).
\end{equation}
Applying the second statement of Proposition \ref{prophpi} shows the
second term in \eqref{mhatExpanDiffTerms} is equal to
\begin{equation} \label{mhatExpanResult2}
-\frac{\pi(x)}{1 - \pi(x)}
 \frac{Q_u(t\,|\,x)}{g(x)na_n} \sj
 K\bigg(\frac{x - X_j}{a_n}\bigg)\zeta(x,Z_j,\delta_j,\tau_0)
 + O\big((na_n)^{-3/4}\log^{3/4}(n)\big).
\end{equation}
Note, the order terms in the displays above hold uniformly over
$(t,\,x) \in (-\infty,\,\tau_0] \times [0,\,1]$. The result then
follows by combining \eqref{mhatExpanResult1} and
\eqref{mhatExpanResult2} with the approximation
\eqref{mhatExpanDiffTerms} under the integral in
\eqref{mhatdiffmremainder}.
\end{proof}

\begin{proof}[{\sc Proof of Proposition \ref{propshat}}]
Beginning with the first assertion, write $\hat v(x) - v(x)$ as
\begin{equation} \label{vhatminusvdef}
\int_0^1 \hxi^2\big(\big(1 - \hpi(x)\big)p\,|\,x\big)J(p)\,dp
 - \int_0^1 \xi^2\big((1 - \pi(x))p\,|\,x\big)J(p)\,dp
 - \mhat^2(x) + m^2(x).
\end{equation}
Since $\mhat^2(x) - m^2(x) = 2m(x)\{\mhat(x) - m(x)\} + \{\mhat(x) -
m(x)\}^2$, it follows from the first statement of Proposition
\ref{propmhat} that $\|\mhat^2 - m^2\|_{\infty} =
O((na_n)^{-1/2}\log^{1/2}(n))$, almost surely. With some technical
effort, the difference of integrals in \eqref{vhatminusvdef} can be
shown to be equal to
\begin{equation*}
\int_0^{\tau_0} \bigg\{ I\bigg( \frac{Q(\sqrt{t}\,|\,x)}{1 - \pi(x)} \bigg)
 - I\bigg( \frac{\hQ(\sqrt{t}\,|\,x)}{1 - \hpi(x)} \bigg) \bigg\}\,dt
 - \int_0^{\infty} \bigg\{
 I\bigg( \frac{Q(-\sqrt{t}\,|\,x)}{1 - \pi(x)} \bigg)
 - I\bigg( \frac{\hQ(-\sqrt{t}\,|\,x)}{1 - \hpi(x)} \bigg) \bigg\}\,dt,
\end{equation*}
where $I(q) = \int_0^q J(p)\,dp$, $q \in [0,\,1]$. It therefore
follows from similar lines of argument to those in the proof of
Proposition \ref{propmhat} that the difference of integrals in
\eqref{vhatminusvdef} is of the order $O((na_n)^{-1/2}\log^{1/2}(n))$,
almost surely, uniformly in $x \in [0,\,1]$. Combining this statement
with the statement $\|\mhat^2 - m^2\|_{\infty} =
O((na_n)^{-1/2}\log^{1/2}(n))$, almost surely, from above, we can see
that $\|\hat v - v\|_{\infty} = O((na_n)^{-1/2}\log^{1/2}(n))$, almost
surely. The first assertion then follows
from the fact that
\begin{equation} \label{shatapprox}
\shat(x) - s(x) + \frac1{2s(x)}\big\{\shat(x) - s(x)\big\}^2
 = \frac1{2s(x)}\big\{\hat v(x) - v(x)\big\}.
\end{equation}

To prove the second assertion, combine \eqref{vhatminusvdef} with
\eqref{shatapprox} to obtain the approximation for $\shat(x) - s(x)$:
\begin{align} \label{shatexpanterms}
&\frac1{2s(x)}\bigg\{
 \int_0^1 \hxi^2\big(\big(1 - \hpi(x)\big)p\,|\,x\big)J(p)\,dp
 - \int_0^1 \xi^2\big((1 - \pi(x))p\,|\,x\big)J(p)\,dp \bigg\}
 \\ \nonumber
&\quad - \frac{m(x)}{s(x)}\big\{\mhat(x) - m(x)\big\}
 + O\big((na_n)^{-1}\log(n)\big),
\end{align}
where we have again used that $\mhat^2(x) - m^2(x) = 2m(x)\{\mhat(x) -
m(x)\} + \{\mhat(x) - m(x)\}^2$ and the order term holds almost
surely, uniformly in $x \in [0,\,1]$. As above, with additional
technical effort, the difference of integrals in
\eqref{shatexpanterms} can be shown to be equal to
\begin{equation*}
2\int_{-\infty}^{\tau_0} \bigg\{
 I\bigg( \frac{Q(y\,|\,x)}{1 - \pi(x)} \bigg)
 - I\bigg( \frac{\hQ(y\,|\,x)}{1 - \hpi(x)} \bigg) \bigg\} y\,dy.
\end{equation*}
Therefore, one can then work with an expansion similar to
\eqref{mhatdiffmremainder} in the proof of Proposition \ref{propmhat}
to derive an approximation of the first term in \eqref{shatexpanterms},
i.e.\
\begin{equation} \label{shatexpanfirstterm}
\int_{-\infty}^{\tau_0} \bigg\{
 \frac{Q(y\,|\,x) - \hQ(y\,|\,x)}{1 - \pi(x)}
 - Q_u(y\,|\,x)\frac{\hpi(x) - \pi(x)}{1 - \pi(x)} \bigg\}
 J\big(Q_u(y\,|\,x)\big)\frac{y}{s(x)}\,dy
 + O\big((na_n)^{-1}\log(n)\big),
\end{equation}
where the order term holds almost surely, uniformly in $x \in
[0,\,1]$. The result then follows by combining the approximations
\eqref{mhatExpanResult1} and \eqref{mhatExpanResult2} from the proof
of Proposition \ref{propmhat} with \eqref{shatexpanfirstterm} for the
first term in \eqref{shatexpanterms} and applying the second statement
of Proposition \ref{propmhat} to the second term in
\eqref{shatexpanterms}.
\end{proof}

With the asymptotic properties of $\hQ$, $\hpi$, $\mhat$ and $\shat$
fully described, we can state the proof of our first main result: a
strong uniform representation for the difference $\hFemp - F$.

\begin{proof}[{\sc Proof of Theorem \ref{thmFhat}}]
We can write $\hFemp(t) - F(t)$ as
\begin{align*}
&\avj \frac{\hQ(t\shat(X_j) + \mhat(X_j)\,|\,X_j)}{1 - \hpi(X_j)}
 - \avj \frac{Q(ts(X_j) + m(X_j)\,|\,X_j)}{1 - \pi(X_j)} \\
&= \avj \frac{\hQ(ts(X_j) + m(X_j)\,|\,X_j)
 - Q(ts(X_j) + m(X_j)\,|\,X_j)}{1 - \pi(X_j)} \\
&\quad + F(t) \avj \frac{\hpi(X_j) - \pi(X_j)}{1 - \pi(X_j)}
 + f(t)\avj \bigg\{ t\frac{\shat(X_j) - s(X_j)}{s(X_j)}
 + \frac{\mhat(X_j) - m(X_j)}{s(X_j)} \bigg\} \\
&\quad + \avj \frac{\hQ(t\shat(X_j) + \mhat(X_j)}{1 - \pi(X_j)}
 \frac{\{\hpi(X_j) - \pi(X_j)\}^2}{1 - \hpi(X_j)} \\
&\quad + \avj \bigg\{ \Big\{
 \hQ\big(t\shat(X_j) + \mhat(X_j)\,\big|\,X_j\big)
 - Q\big(t\shat(X_j) + \mhat(X_j)\,\big|\,X_j\big) \\
&\qquad\phantom{ + \avj \bigg\{ \Big\{}
 - \hQ(ts(X_j) + m(X_j)\,|\,X_j)
 + Q(ts(X_j) + m(X_j)\,|\,X_j) \Big\}
 \Big\slash \Big\{ 1 - \pi(X_j) \Big\} \bigg\} \\
&\qquad\phantom{ + \avj} \times
 \bigg\{ \frac{\hpi(X_j) - \pi(X_j)}{1 - \pi(X_j)} \bigg\} \\
&\quad + \avj \frac{\hQ(ts(X_j) + m(X_j)\,|\,X_j)
 - Q(ts(X_j) + m(X_j)\,|\,X_j)}{1 - \pi(X_j)}
 \frac{\hpi(X_j) - \pi(X_j)}{1 - \pi(X_j)} \\
&\quad + \avj \bigg\{
 F\bigg(t + t\frac{\shat(X_j) - s(X_j)}{s(X_j)}
 + \frac{\mhat(X_j) - m(X_j)}{s(X_j)} \bigg)
 - F(t) \bigg\}
 \frac{\hpi(X_j) - \pi(X_j)}{1 - \pi(X_j)} \\
&\quad + \avj \Big\{
 \hQ\big(t\shat(X_j) + \mhat(X_j)\,\big|\,X_j\big)
 - Q\big(t\shat(X_j) + \mhat(X_j)\,\big|\,X_j\big) \\
&\qquad\phantom{ + \avj \Big\{}
 - \hQ(ts(X_j) + m(X_j)\,|\,X_j)
 + Q(ts(X_j) + m(X_j)\,|\,X_j) \Big\}
 \Big\slash \Big\{1 - \pi(X_j) \Big\} \\
&\quad + \avj \bigg\{
 F\bigg(t + t\frac{\shat(X_j) - s(X_j)}{s(X_j)}
 + \frac{\mhat(X_j) - m(X_j)}{s(X_j)}\bigg) - F(t) \\
&\qquad\phantom{ + \avj \bigg\{}
 - f(t)\bigg\{t\frac{\shat(X_j) - s(X_j)}{s(X_j)}
 + \frac{\mhat(X_j) - m(X_j)}{s(X_j)} \bigg\} \bigg\}.
\end{align*}
We can therefore write
\begin{equation*}
R_{4,n}(t) = \hFemp(t) - F(t) - E_n(t)
 = \sum_{i = 1}^{10} D_{i,n}(t),
\end{equation*}
where $E_n(t)$ is defined in the statement of the theorem and
\begin{equation*}
D_{1,n}(t) = \avj \frac{\hQ(t\shat(X_j) + \mhat(X_j)\,|\,X_j)}{1 - \pi(X_j)}
 \frac{\{\hpi(X_j) - \pi(X_j)\}^2}{1 - \hpi(X_j)},
\end{equation*}
\begin{align*}
D_{2,n}(t) = \avj &\bigg\{\Big\{
 \hQ\big(t\shat(X_j) + \mhat(X_j)\,\big|\,X_j\big)
 - Q\big(t\shat(X_j) + \mhat(X_j)\,\big|\,X_j\big) \\
&\phantom{\bigg\{\Big\{}\quad - \hQ(ts(X_j) + m(X_j)\,|\,X_j)
 + Q(ts(X_j) + m(X_j)\,|\,X_j) \Big\}
 \Big\slash \Big\{ 1 - \pi(X_j) \Big\}\bigg\} \\
&\quad\times \bigg\{ \frac{\hpi(X_j) - \pi(X_j)}{1 - \pi(X_j)} \bigg\},
\end{align*}
\begin{equation*}
D_{3,n}(t) = \avj \frac{\hQ(ts(X_j) + m(X_j)\,|\,X_j)
 - Q(ts(X_j) + m(X_j)\,|\,X_j)}{1 - \pi(X_j)}
 \frac{\hpi(X_j) - \pi(X_j)}{1 - \pi(X_j)},
\end{equation*}
\begin{equation*}
D_{4,n}(t) = \avj \bigg\{F\bigg(t + t\frac{\shat(X_j) - s(X_j)}{s(X_j)}
 + \frac{\mhat(X_j) - m(X_j)}{s(X_j)}\bigg) - F(t)\bigg\}
 \frac{\hpi(X_j) - \pi(X_j)}{1 - \pi(X_j)},
\end{equation*}
\begin{align*}
D_{5,n}(t) = \avj \Big\{& \hQ\big(t\shat(X_j) + \mhat(X_j)\,\big|\,X_j\big)
 - Q\big(t\shat(X_j) + \mhat(X_j)\,\big|\,X_j\big) \\
&\quad - \hQ(ts(X_j) + m(X_j)\,|\,X_j) + Q(ts(X_j) + m(X_j)\,|\,X_j) \Big\}
 \Big\slash \Big\{ 1 - \pi(X_j) \Big\},
\end{align*}
\begin{align*}
D_{6,n}(t) = \avj \bigg\{&
 F\bigg(t + t\frac{\shat(X_j) - s(X_j)}{s(X_j)}
 + \frac{\mhat(X_j) - m(X_j)}{s(X_j)}\bigg) - F(t) \\
&\quad - f(t) \bigg\{ t\frac{\shat(X_j) - s(X_j)}{s(X_j)}
 + \frac{\mhat(X_j) - m(X_j)}{s(X_j)} \bigg\} \bigg\},
\end{align*}
\begin{equation*}
D_{7,n}(t) = \avj \frac{R_n(ts(X_j) + m(X_j)\,|\,X_j)}{1 - \pi(X_j)},
\end{equation*}
where $R_n$ is given in the third statement of Lemma
\ref{lemhQConsistencyModulusExpan},
\begin{equation*}
D_{8,n}(t) = F(t)\avj \frac{R_{1,n}(X_j)}{1 - \pi(X_j)},
\end{equation*}
where $R_{1,n}$ is given in the second statement of Proposition
\ref{prophpi},
\begin{equation*}
D_{9,n}(t) = f(t)\avj \frac{R_{2,n}(X_j)}{s(X_j)},
\end{equation*}
where $R_{2,n}$ is given in the second statement of Proposition
\ref{propmhat}, and
\begin{equation*}
D_{10,n}(t) = tf(t)\avj \frac{R_{3,n}(X_j)}{s(X_j)},
\end{equation*}
where $R_{3,n}$ is given in the second statement of Proposition
\ref{propshat}.

The assumptions of Proposition \ref{prophpi} are satisfied, and the
first statement of this result gives $\|\hpi - \pi\|_{\infty} =
O((na_n)^{-1/2}\log^{1/2}(n))$, almost surely. Combining this
statement with the facts that $\|\pi\|_{\infty} < 1$ and that
$|D_{1,n}(t)|$ is bounded by $[1 - \|\pi\|_{\infty}]^{-1} [1 -
\|\hpi\|_{\infty}]^{-1} \|\hpi - \pi\|_{\infty}^2$ shows that
$\|D_{1,n}\|_{\infty}$ is of the order $O((na_n)^{-1}\log(n)) =
o((na_n)^{-3/4}\log^{3/4}(n))$, almost surely.

We can see that $|D_{2,n}(t)|$ is bounded by
\begin{align} \label{D2nBound}
&\sup_{x \in [0,\,1]} \sup_{-\infty < t \leq \tau_F} \Big|
 \hQ\big(t\shat(x) + \mhat(x)\,\big|\,x\big)
 - Q\big(t\shat(x) + \mhat(x)\,\big|\,x\big) \\ \nonumber
&\phantom{\sup_{x \in [0,\,1]} \sup_{-\infty < t \leq \tau_F} \Big|}\quad
 - \hQ(ts(x) + m(x)\,|\,x) + Q(ts(x) + m(x)\,|\,x) \Big| \\ \nonumber
&\times \Big[1 - \|\pi\|_{\infty}\Big]^2
 \sup_{x \in [0,\,1]} \Big| \hpi(x) - \pi(x) \Big|,
\end{align}
and we have already used that $\|\pi\|_{\infty} < 1$ and $\|\hpi -
\pi\|_{\infty} = O((na_n)^{-1/2}\log^{1/2}(n))$, almost surely. Hence,
we only need to treat the first term in \eqref{D2nBound}, which we do
by using the modulus of continuity for the Beran estimator $\hQ$ given
in the second statement of Lemma
\ref{lemhQConsistencyModulusExpan}. The assumptions of Lemma
\ref{lemhQConsistencyModulusExpan} are satisfied. However, in order to
use the second statement of this result, we need to show the related
result
\begin{equation*}
\sup_{x \in [0,\,1]} \sup_{-\infty < t \leq \tau_F} \Big|
L\big(t\shat(x) + \mhat(x)\big) - L\big(ts(x) + m(x)\big) \Big|
 = O\big((na_n)^{-1/2}\log^{1/2}(n)\big),
 \qquad\text{a.s.}
\end{equation*}
This is equivalent to showing
\begin{equation} \label{D2nModulusPart}
\sup_{x \in [0,\,1]} \sup_{-\infty < t \leq \tau_0} \bigg|
 L\bigg(t + t\frac{\shat(x) - s(x)}{s(x)}
 + \frac{\mhat(x) - m(x)}{s(x)}\bigg) - L(t) \bigg|
 = O\big((na_n)^{-1/2}\log^{1/2}(n)\big),
 \qquad\text{a.s.}
\end{equation}

To see that \eqref{D2nModulusPart} holds, recall the sequences of real
numbers $\{u_n\}_{n \geq 1}$ and $\{v_n\}_{n \geq 1}$ introduced in
the discussion following \eqref{assumpLcurve}. It follows that $|L(t +
t(v_n - v) + u_n - u) - L(t)|$ is bounded by
\begin{align} \label{LBound}
&\big|u_n - u\big|\sup_{-\infty < t \leq \tau_0} \Big|L'(t)\Big|
 + \big|v_n - v\big|\sup_{-\infty < t \leq \tau_0} \Big|tL'(t)\Big|
 \\ \nonumber
&+ \sup_{-\infty < t \leq \tau_0} \Big|
 L\big(t + t(v_n - v) + u_n - u\big) - L\big(t + t(v_n - v)\big)
 - (u_n - u)L'\big(t + t(v_n - v)\big) \Big|
 \\ \nonumber
&+ \sup_{-\infty < t \leq \tau_0} \Big|
 L\big(t + t(v_n - v)\big) - L(t)
 - (v_n - v)tL'(t) \Big|.
\end{align}
The terms in the first line of \eqref{LBound} are easily seen to be of
the order $O((v_n - v) + (u_n - u))$, as desired. The quantity inside
the absolute brackets in the second line of the same display is equal
to
\begin{equation} \label{LTerm1Holder}
(u_n - u)\int_0^1 \Big\{ L'\big(t + t(v_n - v) + p(u_n - u)\big)
 - L'\big(t + t(v_n - v)\big)\Big\}\,dp.
\end{equation}

We will now use \eqref{assumpLcurve} to show that \eqref{LTerm1Holder}
is of the order $O(|u_n - u|^2)$. Let $-\infty < a < b \leq
\tau_0$. It follows from the fact that $L$ is nondecreasing that $(1 +
a^2)|L'(b) - L'(a)|$ is bounded by
\begin{equation*}
\Big|\big(1 + b^2\big)L'(b) - \big(1 + a^2\big)L'(a)\Big| =
 \bigg|\int_a^b\,\big(1 + w^2\big)L''(w)\,dw
 + 2\int_a^b\,wL'(w)\,dw\bigg|.
\end{equation*}
The triangle inequality in combination with \eqref{assumpLcurve}
implies that the right-hand side of the display above is further
bounded by
\begin{align*}
&\int_a^b\,\big(1 + w^2\big)\bigg|\frac{L''(w)}{L'(w)}\bigg|\,L(dw)
 + 2\sup_{-\infty < t \leq \tau_0}\big|tL'(t)\big| \big|b - a\big| \\
&\leq \Bigg\{\bigg\{\int_0^1\,\Big(1 + \big(a + p(b - a)\big)^2\Big)
 \bigg\{\frac{L''(a + p(b - a))}{L'(a + p(b - a))}\bigg\}^2
 L'\big(a + p(b - a)\big)\,dp \\
&\quad\phantom{\Bigg\{\bigg\{} \times
 \int_0^1\,\Big(1 + \big(a + p(b - a)\big)^2\Big)
 L'\big(a + p(b - a)\big)\,dp\bigg\}^{1/2}
 + 2\sup_{-\infty < t \leq \tau_0}\Big|tL'(t)\Big|\Bigg\}\big|b - a\big| \\
&\leq \Bigg\{\bigg\{ \int_{-\infty}^{\tau_0}\,
 \big(1 + w^2\big)\bigg\{\frac{L''(w)}{L'(w)}\bigg\}\,L(dw)
 \int_{-\infty}^{\tau_0}\,\big(1 + w^2\big)\,L(dw)\bigg\}^{1/2}
 + 2\sup_{-\infty < t \leq \tau_0}\Big|tL'(t)\Big|\Bigg\}\big|b - a\big|,
\end{align*}
where the middle inequality follows from H\"older's
inequality and the final inequality follows by the facts that the
integrands are nonnegative and $[a,\,b] \subset
(-\infty,\,\tau_0]$. This means that we can find a constant $C > 0$
such that
\begin{equation} \label{LLipschitz}
\Big|L'(b) - L'(a)\Big| \leq C\frac{|b - a|}{1 + a^2},
\qquad -\infty < a, b \leq \tau_0.
\end{equation}

Setting $a = \min\{t + t(v_n - v),\,t + t(v_n - v) + p(u_n - u)\}$ and
$b = \max\{t + t(v_n - v),\,t + t(v_n - v) + p(u_n - u)\}$ in
\eqref{LLipschitz} implies that $|L'(t + t(v_n - v) + p(u_n - u)) -
L'(t + t(v_n - v))|$ is bounded by
\begin{equation} \label{LTerm1Bound}
C\sup_{-\infty < t \leq \tau_0}\sup_{0 \leq p \leq 1}
 \frac{|p|}{1 + (\min\{t + t(v_n - v),\,t + t(v_n - v) + p(u_n - u)\})^2}
 \big|u_n - u\big|.
\end{equation}
It follows for \eqref{LTerm1Holder} to be of the order $O(|u_n -
u|^2)$, and, hence, the term in the second line of \eqref{LBound} is
also of the order $O(|u_n - u|^2) = o(|u_n - u|)$. A similar argument
shows the term in the third line of \eqref{LBound} is of the order
$O(|v_n - v|^2) = o(|v_n - v|)$, where the fraction in
\eqref{LLipschitz} becomes $|p|t^2/(1 + t^2(\min\{1,\,1 + v_n -
v\})^2)$, which is bounded for all $n$ where $-1 \neq v_n - v$. This
shows the desired result
\begin{equation} \label{LOrder}
\sup_{-\infty < t \leq \tau_0} \Big|L(t + t(v_n - v) + (u_n - u)) - L(t)\Big|
 = O\big(|v_n - v| + |u_n - u|\big).
\end{equation}

Since the assumptions of Proposition \ref{propmhat} and Proposition
\ref{propshat} are satisfied, combining the first statements of these
results with the fact that $s$ is bounded away from zero and
\eqref{LOrder} establishes the desired \eqref{D2nModulusPart}. We can
therefore apply the second statement of Lemma
\ref{lemhQConsistencyModulusExpan} to see that the first term of
\eqref{D2nBound} is of the order $O((na_n)^{-3/4}\log^{3/4}(n))$. It
then follows that $\|D_{2,n}\|_{\infty}$ is of the order
$O((na_n)^{-5/4}\log^{5/4}(n)) = o((na_n)^{-3/4}\log^{3/4}(n))$,
almost surely.

We can use the first statements of Lemma
\ref{lemhQConsistencyModulusExpan} and Proposition \ref{prophpi} to
treat $D_{3,n}(t)$, since this remainder term is bounded in absolute
value by
\begin{equation*}
\Big[1 - \|\pi\|_{\infty}\Big]^2
 \sup_{x \in [0,\,1]}\sup_{-\infty < t \leq \tau_0}
 \big|\hQ(t\,|\,x) - Q(t\,|\,x)\big|
 \sup_{x \in [0,\,1]} \big|\hpi(x) - \pi(x)\big|.
\end{equation*}
Therefore, $\|D_{3,n}\|_{\infty}$ is of the order
$O((na_n)^{-1}\log(n)) = o((na_n)^{-3/4}\log^{3/4}(n))$, almost surely.

Since $F$ satisfies \eqref{assumpLcurve}, with $F$ in place of $L$,
$f$ in place of $L'$ and $f'$ in place of $L''$, the same argument
used to verify \eqref{D2nModulusPart} can be used to show
\begin{equation*}
\sup_{x \in [0,\,1]}\sup_{-\infty < t \leq \tau_F} \bigg|
 F\bigg(t + t\frac{\shat(x) - s(x)}{s(x)}
 + \frac{\mhat(x) - m(x)}{s(x)}\bigg) - F(t) \bigg|
 = O\big((na_n)^{-1/2}\log^{1/2}(n)\big),
 \qquad\text{a.s.}
\end{equation*}
This fact combined with the result $\|\hpi - \pi\|_{\infty} =
O((na_n)^{-1/2}\log^{1/2}(n))$, almost surely, from the first
statement of Proposition \ref{prophpi}, and the fact that
$|D_{4,n}(t)|$ is bounded by
\begin{equation*}
\Big[1 - \|\pi\|_{\infty}\Big]^{-1}
 \sup_{x \in [0,\,1]}\sup_{-\infty < t \leq \tau_F} \bigg|
 F\bigg(t + t\frac{\shat(x) - s(x)}{s(x)}
 + \frac{\mhat(x) - m(x)}{s(x)}\bigg) - F(t) \bigg|
 \sup_{x \in [0,\,1]} \Big|\hpi(x) - \pi(x)\Big|
\end{equation*}
shows that $\|D_{4,n}\|_{\infty}$ is of the order
$O((na_n)^{-1}\log(n)) = o((na_n)^{-3/4}\log^{3/4}(n))$, almost
surely.

Similar to the arguments for the remainder term $D_{2,n}(t)$, we can
apply the second statement of Lemma
\ref{lemhQConsistencyModulusExpan} and the fact that $\|\pi\|_{\infty}
< 1$ to treat the remainder term $D_{5,n}(t)$, since $|D_{5,n}(t)|$ is
bounded by $[1 - \|\pi\|_{\infty}]^{-1}$ multiplied by
\begin{equation*}
\sup_{x \in [0,\,1]} \sup_{-\infty < t \leq \tau_F} \Big|
 \hQ\big(t\shat(x) + \mhat(x)\,\big|\,x\big)
 - Q\big(t\shat(x) + \mhat(x)\,\big|\,x\big)
 - \hQ(ts(x) + m(x)\,|\,x) + Q(ts(x) + m(x)\,|\,x) \Big|.
\end{equation*}
Therefore, $\|D_{5,n}\|_{\infty}$ is of the order
$O((na_n)^{-3/4}\log^{3/4}(n))$, almost surely, because we have
already shown that the quantity in the display above is
$O((na_n)^{-3/4}\log^{3/4}(n))$, almost surely, using the modulus of
continuity of the Beran estimator $\hQ$ given in the second statement
of Lemma \ref{lemhQConsistencyModulusExpan}.

We can write the remainder term $D_{6,n}(t)$ as the sum
\begin{align*}
&\avj \bigg\{ F\bigg(t + t\frac{\shat(X_j) - s(X_j)}{s(X_j)}
 + \frac{\mhat(X_j) - m(X_j)}{s(X_j)}\bigg)
 - F\bigg(t + t\frac{\shat(X_j) - s(X_j)}{s(X_j)}\bigg) \\
&\phantom{\avj \bigg\{} \quad
 - f\bigg(t + t\frac{\shat(X_j) - s(X_j)}{s(X_j)}\bigg)
 \frac{\mhat(X_j) - m(X_j)}{s(X_j)}\bigg\} \\
& + \avj \bigg\{ F\bigg(t + t\frac{\shat(X_j) - s(X_j)}{s(X_j)}\bigg)
 - F(t) - tf(t)\frac{\shat(X_j) - s(X_j)}{s(X_j)} \bigg\} \\
& + \avj \bigg\{ f\bigg(t + t\frac{\shat(X_j) - s(X_j)}{s(X_j)}\bigg)
 - f(t) \bigg\} \frac{\mhat(X_j) - m(X_j)}{s(X_j)}.
\end{align*}
Hence, $|D_{6,n}(t)|$ is bounded by the sum of the quantities
\begin{align} \label{D6Bound1}
\sup_{x \in [0,\,1]} \sup_{-\infty < t \leq \tau_F} \bigg|
& F\bigg(t + t\frac{\shat(x) - s(x)}{s(x)}
 + \frac{\mhat(x) - m(x)}{s(x)}\bigg)
 - F\bigg(t + t\frac{\shat(x) - s(x)}{s(x)}\bigg)
 \\ \nonumber
&\quad - f\bigg(t + t\frac{\shat(x) - s(x)}{s(x)}\bigg)
 \frac{\mhat(x) - m(x)}{s(x)} \bigg|,
\end{align}
\begin{equation} \label{D6Bound2}
\sup_{x \in [0,\,1]} \sup_{-\infty < t \leq \tau_F} \bigg|
 F\bigg(t + t\frac{\shat(x) - s(x)}{s(x)}\bigg)
 - F(t) - tf(t)\frac{\shat(x) - s(x)}{s(x)} \bigg|
\end{equation}
and
\begin{equation} \label{D6Bound3}
\bigg[\inf_{x \in [0,\,1]} s(x)\bigg]^{-1}
\sup_{x \in [0,\,1]} \sup_{-\infty < t \leq \tau_F} \bigg|
 f\bigg(t + t\frac{\shat(x) - s(x)}{s(x)}\bigg)
 - f(t) \bigg|
\sup_{x \in [0,\,1]} \Big|\mhat(x) - m(x)\Big|.
\end{equation}
Since $F$ satisfies \eqref{assumpLcurve}, analogous arguments to those
that are used to show the second and third terms of \eqref{LBound} are
of the orders $O(|u_n - u|^2)$ and $O(|v_n - v|^2)$, respectively,
combined with the first statements of Proposition \ref{propmhat} and
Proposition \ref{propshat} show the bounds \eqref{D6Bound1} and
\eqref{D6Bound2} are both of the order $O((na_n)^{-1}\log(n)) =
o((na_n)^{-3/4}\log^{3/4}(n))$, almost surely. Also, a similar
argument that is used to find the bound \eqref{LTerm1Bound} combined with
the first statement of Proposition \ref{propshat} can be used to show
the second term in \eqref{D6Bound3} is of the order
$O((na_n)^{-1/2}\log^{1/2}(n))$, almost surely. The third term in
\eqref{D6Bound3} has the order $O((na_n)^{-1/2}\log^{1/2}(n))$, almost
surely, from the first statement of Proposition
\ref{propmhat}. Therefore, we can see that $\|D_{6,n}\|_{\infty}$ is
of the order $O((na_n)^{-1}\log(n)) = o((na_n)^{-3/4}\log^{3/4}(n))$,
almost surely.

The assumptions of Lemma \ref{lemhQConsistencyModulusExpan} are
satisfied, and it follows from the third statement of this result that
$\|R_n\|_{\infty}$ is of the order $O((na_n)^{-3/4}\log^{3/4}(n))$,
almost surely. It then follows that $\|D_{7,n}\|_{\infty}$ is also of
the order $O((na_n)^{-3/4}\log^{3/4}(n))$, almost surely, because
$|D_{7,n}(t)|$ is bounded by $[1 -
\|\pi\|_{\infty}]^{-1}\|R_n\|_{\infty}$. Since $|D_8(t)|$ is bounded
by $[1 - \|\pi\|_{\infty}]^{-1}\|R_{1,n}\|_{\infty}$, it follows from
the second statement of Proposition \ref{prophpi} for
$\|D_8\|_{\infty}$ to be of the order $O((na_n)^{-3/4}\log^{3/4}(n))$,
almost surely. The second statement from
Proposition \ref{propmhat} shows that $\|D_9\|_{\infty}$ is of the
order $O((na_n^{-3/4}\log^{3/4}(n))$, almost surely, which follows
from the fact that $|D_9(t)|$ is bounded by
$\|f\|_{\infty}[\inf_{x \in [0,\,1]}
s(x)]^{-1}\|R_{2,n}\|_{\infty}$. Similarly, the second statement of
Proposition \ref{propshat} shows that $\|D_{10,n}\|_{\infty}$ is of
the order $O((na_n)^{-3/4}\log^{3/4}(n))$, almost surely, which
concludes the proof.
\end{proof}

Before we can prove of our second main result we need to state the
asymptotic order of the mean of the process $\{E_n(t)\,:\,-\infty < t
\leq \tau_F\}$ introduced in Theorem \ref{thmFhat}.

\begin{lemma} \label{lemEntmeanOrder}
Under the conditions of Theorem \ref{thmFhat} it follows that
\begin{equation*}
\sup_{-\infty < t \leq \tau_F} \Big| E\big[E_n(t)\big] \Big|
 = O\big(a_n^2\big).
\end{equation*}
\end{lemma}
\begin{proof}
Recall from the statement of Theorem \ref{thmFhat} that $E_n(t) =
T_n(t) - f(t)\{U_n + tV_n\} - W_n(t)$. Hence, the assertion follows
from showing $\|E[T_n]\|_{\infty} = O(a_n^2)$, $|E[U_n]| = O(a_n^2)$,
$|E[V_n]| = O(a_n^2)$ and $\|E[W_n]\|_{\infty} = O(a_n^2)$. We will
only show the result that $\|E[T_n]\|_{\infty} = O(a_n^2)$ because the
remaining statements follow by similar lines of argument.

Write
\begin{align*}
&\frac{1 - Q(ts(v + a_nw) + m(v + a_nw)\,|\,v + a_mw)}{1 - \pi(v + a_nw)}
 - \frac{1 - Q(ts(v) + m(v)\,|\,v)}{1 - \pi(v)} \\
&= a_nw\frac{\int_0^1\,\dot{\pi}(v + a_nwp)\,dp}{
 \{1 - \pi(v + a_nw)\}\{1 - \pi(v)\}},
\end{align*}
where $\dot{\pi}$ is the first derivative of the function $\pi$
with respect to its argument. Additionally, write
\begin{align} \label{ETnSecondTermApprox}
&E\Big[ \zeta\big(v + a_nw, Z, \delta, ts(v + a_nw) + m(v + a_nw)\big)
 \,\Big|\,X = v,\,X' = v + a_nw\Big] \\ \nonumber
&= -a_nw \int_{-\infty}^{ts(v) + m(v)}\,
 \frac1{1 - M(s-\,|\,v)}\,\dot{M}^1(ds\,|\,v)
 - a_nw \int_{-\infty}^{ts(v) + m(v)}\,
 \frac{\dot{M}(s-\,|\,v)}{\{1 - M(s-\,|\,v)\}^2}\,
 M^1(ds\,|\,v) \\ \nonumber
&\quad + a_nw \int_{ts(v) + m(v)}^{ts(v + a_nw) + m(v + a_nw)}\,
 \frac{\int_0^1\,\dot{M}(s-\,|\,v + a_nwp)\,dp}{
 \{1 - M(s-\,|\,v + a_nw)\}\{1 - M(s-\,|\,v)\}}\,
 M^1(ds\,|\,v) \\ \nonumber
&\quad - a_nw \int_{ts(v) + m(v)}^{ts(v + a_nw) + m(v + a_nw)}\,
 \frac{\int_0^1\,\dot{M}(s-\,|\,v + a_nwp)\,dp}{
 \{1 - M(s-\,|\,v + a_nw)\}^2}\,
 M^1(ds\,|\,v + a_nw) \\ \nonumber
&\quad - a_nw \int_{ts(v) + m(v)}^{ts(v + a_nw) + m(v + a_nw)}\,
 \frac{\int_0^1\,\dot{M}(s-\,|\,v + a_nwp)\,dp}{
 \{1 - M(s-\,|\,v + a_nw)\}\{1 - M(s-\,|\,v)\}}\,
 M^1(ds\,|\,v + a_nw) \\ \nonumber
&\quad - a_nw \int_0^1\, \bigg\{ \int_{ts(v) + m(v)}^{ts(v + a_nw) + m(v + a_nw)}\,
 \frac1{1 - M(s-\,|\,v)}\,\dot{M}^1(ds\,|\,v + a_nwq)
 \bigg\}\,dq \\ \nonumber
&\quad - a_n^2w^2 \int_{-\infty}^{ts(v) + m(v)}\,
 \frac{\int_0^1\,p\{\int_0^1\,\ddot{M}(s-\,|\,v + a_nwpq)\,dq\}\,dp}{
 \{1 - M(s-\,|\,v)\}^2}\,
 M^1(ds\,|\,v) \\ \nonumber
&\quad - a_n^2w^2 \int_{-\infty}^{ts(v) + m(v)}\,
 \frac{\{\int_0^1\,\dot{M}(s-\,|\,v + a_nwp)\,dp\}^2}{
 \{1 - M(s-\,|\,v + a_nw)\}\{1 - M(s-\,|\,v)\}^2}\,
 M^1(ds\,|\,v) \\ \nonumber
&\quad - a_n^2w^2 \int_{-\infty}^{ts(v) + m(v)}\,
 \frac{\{\int_0^1\,\dot{M}(s-\,|\,v + a_nwp)\,dp\}^2}{
 \{1 - M(s-\,|\,v + a_nw)\}^2\{1 - M(s-\,|\,v)\}}\,
 M^1(ds\,|\,v) \\ \nonumber
&\quad - a_n^2w^2 \int_0^1\, \bigg\{ \int_{-\infty}^{ts(v) + m(v)}\,
 \frac{\int_0^1\,\dot{M}(s-\,|\,v + a_nwp)\,dp}{
 \{1 - M(s-\,|\,v + a_nw)\}^2}\,
 \dot{M}^1(ds\,|\,v + a_nwq) \bigg\}\,dq \\ \nonumber
&\quad - a_n^2w^2 \int_0^1 \bigg\{ \int_{-\infty}^{ts(v) + m(v)}\,
 \frac{\int_0^1\,\dot{M}(s-\,|\,v + a_nwp)\,dp}{
 \{1 - M(s-\,|\,v + a_nw)\}\{1 - M(s-\,|\,v)\}}\,
 \dot{M}^1(ds\,|\,v + a_nwq) \bigg\}\,dq \\ \nonumber
&\quad - a_n^2w^2 \int_0^1\, p\bigg\{ \int_0^1\, \bigg\{
 \int_{-\infty}^{ts(v) + m(v)}\,\frac1{1 - M(s-\,|\,v)}\,
 \ddot{M}^1(ds\,|\,v + a_nwpq) \bigg\}\,dq \bigg\}\,dp.
\end{align}
Here $\ddot{M}$ is the second partial derivative of $M$ with respect to
$x$ and $\ddot{M}^1$ is the second partial derivative of $M^1$ with
respect to $x$. For large enough $n$, $E[T_n(t)]$ is equal to
\begin{align} \label{ETn}
&\int_0^1 \bigg\{ \int_{-1}^{1}\,E\Big[
 \zeta\big(v + a_nw,Z,\delta,ts(v + a_nw) + m(v + a_nw)\big)
 \,\Big|\,X = v,\,X' = v + a_nw \Big]K(w)\,dw\bigg\} \\ \nonumber
&\phantom{\int_0^1}\quad\times
 \frac{1 - Q(ts(v) + m(v)\,|\,v)}{1 - \pi(v)}\,dv \\ \nonumber
& + a_n\int_0^1 \bigg\{ \int_{-1}^{1}\,E\Big[
 \zeta\big(v + a_nw,Z,\delta,ts(v + a_nw) + m(v + a_nw)\big)
 \,\Big|\,X = v,\,X' = v + a_nw \Big] \\ \nonumber
&\phantom{ + a_n\int_0^1 \bigg\{ \int_{-1}^{1}\,}\quad \times
 \frac{\int_0^1\,\dot{\pi}(v + a_nws)\,ds}{
 \{1 - \pi(v + a_nw)\}\{1 - \pi(v)\}}
 wK(w)\,dw \bigg\}\,dv.
\end{align}
Since the first two terms on the right-hand side of
\eqref{ETnSecondTermApprox} depend only on $w$ multiplied by a
quantity not depending on $w$, the kernel function $K$ having mean
zero implies the associated terms in \eqref{ETn} are equal to
zero, while the remaining terms in the right-hand side of
\eqref{ETnSecondTermApprox} are easily shown to be of the order
$O(a_n^2)$. The assertion then follows by combining the right-hand
side of \eqref{ETnSecondTermApprox} with expression \eqref{ETn}, and
observing the remaining nonzero terms are all of the order $O(a_n^2)$ or
$o(a_n^2)$.
\end{proof}

To continue we will introduce some notation. Write $\mathcal{H}$ for a
class of measurable functions and let $\rho$ be a pseudometric for
$\mathcal{H}$. As is in Definition 2.1.5 of van der Vaart and Wellner
(1996), we will call $N(\epsilon,\mathcal{H},\rho)$ the covering
number of $\mathcal{H}$, which is the minimum number of balls
$\{g\,:\,\rho(g,h) < \epsilon\}$ of radius $\epsilon$ that is
required to cover $\mathcal{H}$. Note that the centers of the balls
need not belong to $\mathcal{H}$, but are required to have finite
length under $\rho$. We will call the logarithm of the covering number
the entropy. Also as in Definition 2.1.6 of van der Vaart and Wellner
(1996), when given two functions satisfying $h_l \leq h_u$ we will
call the collection of functions from $\mathcal{H}$ satisfying $h_l
\leq h \leq h_u$ a bracket, and an $\epsilon$-bracket when
the length of $h_u - h_l$ under $\rho$ is smaller than $\epsilon$. We
will then call the minimum number of $\epsilon$-brackets required to
cover $\mathcal{H}$ the bracketing number of $\mathcal{H}$, and write
$N_{[\,]}(\epsilon,\mathcal{H},\rho)$. As in the definition of the
covering number, the bracketing functions $h_l \leq h_u$ need not
belong to $\mathcal{H}$ but are required to have finite lengths under
$\rho$.

It is common to let the pseudometric $\rho$ be a scaled
$L_q(P)$--metric for some $q \geq 1$ ($\rho_{L_q(P)}(h,g) \propto [\,
\int\,|h - g|^q\,dP\,]^{1/q}$) or a scaled supremum metric (see
$\|\cdot\|_\infty$ introduced at the beginning of Section \nolinebreak
\ref{estimates}). A function $H$ such that $|h| \leq H$ for every $h
\in \mathcal{H}$ is called an envelope function for $\mathcal{H}$, and
this function is useful for scaling the pseudometric $\rho$. When
$\int\,H^q\,dP < \infty$ it is helpful to think of $\mathcal{H}$ as a
subset of the class $\mathcal{L}_q(P)$, writing $\mathcal{L}_q(P)$ for
the class of measurable functions with finite length under the
$L_q(P)$--metric. Covering numbers and bracketing numbers are very
helpful in understanding asymptotic properties of the process
$\{n^{-1/2}\{\hFemp(t) - F(t)\}\,:\, -\infty < t \leq \tau_F\}$, which
depends on the index set $-\infty < t \leq \tau_F$ and a bandwidth
sequence $\{a_n\}_{n \geq 1}$. We conclude this section with the proof
of our second main result: the weak convergence of $n^{1/2}\{\hFemp -
F\}$.

\begin{proof}[{\sc Proof of Theorem \ref{thmFhatWC}}]
The conditions of Theorem \ref{thmFhat} are satisfied with
$(na_n)^{-3/4}\log^{3/4}(n) = o(n^{-1/2})$, and we can
write $\hFemp(t) - F(t) = E_n(t) + \opn$, $-\infty < t \leq \tau_F$,
where the process $E_n(t) = T_n(t) - f(t)\{U_n + tV_n\} - W_n(t)$
depends on the random quantities $T_n(t)$, $U_n$, $V_n$ and $W_n(t)$
given in Theorem \ref{thmFhat}. Since $E_n(t)$ is not centered and the
conditions of Lemma \ref{lemEntmeanOrder} are satisfied with $a_n^2 =
o(n^{-1/2})$, we center the process $E_n(t)$ to obtain $\hFemp(t) -
F(t) = E_n(t) - E[E_n(t)] + \opn = T_n(t) - E[T_n(t)] - f(t)\{U_n -
E[U_n] + t\{V_n - E[V_n]\}\} - \{W_n(t) - E[W_n(t)]\} + \opn$. The
assertion follows if each of $T_n(t) - E[T_n(t)]$, $U_n - E[U_n]$,
$V_n - E[V_n]$ and $W_n(t) - E[W_n(t)]$ are asymptotically linear and
satisfy appropriate central limit theorems. We will prove only that
$T_n(t) - E[T_n(t)]$ is asymptotically linear, uniformly in $-\infty <
t \leq \tau_F$, and satisfies a functional central limit theorem. The
remaining statements can be shown using similar and easier arguments
that have been omitted for brevity.

We will now introduce some notation. As in Pakes and Pollard (1989),
we will call a class of functions $\mathcal{H}$ a {\em Euclidean
class} with envelope function $H$ (with respect to the $L_q(P)$--metric)
when there are constants $C_1,C_2 > 0$ such that the covering numbers
$N(\epsilon,\mathcal{H},L_q(P))$ satisfy
\begin{equation*}
N\big(\epsilon,\mathcal{H},L_q(P)\big) \leq C_2\epsilon^{-C_1},
 \quad 0 < \epsilon \leq 1.
\end{equation*}
The constants $C_1$ and $C_2$ cannot depend on $P$. Note, Sherman
(1994) requires that the envelope $H$ satisfies $\int\,H^2\,dP <
\infty$. This condition is always satisfied for uniformly bounded $H$,
and in this case we do not mention the distribution $P$.

To show that $T_n(t) - E[T_n(t)]$ is asymptotically linear, we will
apply results from Sherman (1994), who studies weak convergence of
degenerate $U$-processes of order $k \geq 1$. Using the Hoeffding
decomposition of a $U$-process, this author is able to obtain several
useful results concerning tightness properties of these
processes. Corollary 7 of Sherman (1994) states that $k$-th order
$U$-processes indexed by a Euclidean class of mean zero functions is
asymptotically tight at the root-$n$ rate, i.e.\ $\Opn$.

The class of mean zero functions associated to $T_n(t) - E[T_n(t)]$ is
$\mathcal{T} = \mathcal{T}_1 - \mathcal{T}_2$ with $\mathcal{T}_1$
equal to
\begin{align*}
\bigg\{ ((X,Z,\delta),\,(X',Z',\delta')) \mapsto
 &\frac1{a}K\bigg(\frac{X - X'}{a}\bigg)
 \frac{1 - Q(ts(X) + m(X)\,|\,X)}{1 - \pi(X)}
 \frac{\zeta(X,Z',\delta',ts(X) + m(X))}{g(X)} \\
&\quad \,:\, 0 < a < 1,\, -\infty < t \leq \tau_F \bigg\}
\end{align*}
and $\mathcal{T}_2 = \{((X,Z,\delta),\,(X',Z',\delta')) \mapsto
E[f_{a,t}((X,Z,\delta),(X',Z',\delta'))]\,:\,f_{a,t} \in
\mathcal{T}_1\}$. We can see that the amount of entropy residing in
the class $\mathcal{T}_2$ is proportional to that residing in the class
$\mathcal{T}_1$, which can be decomposed into the product of three
classes:
\begin{equation*}
\mathcal{K} = \bigg\{ U = X - X' \mapsto
 \frac1{a}K\bigg(\frac{U}{a}\bigg)\,:\,0 < a < 1 \bigg\},
\end{equation*}
\begin{align*}
\mathcal{S} &= \bigg\{ X \mapsto
 \frac{1 - Q(ts(X) + m(X)\,|\,X)}{1 - \pi(X)}
 \,:\, -\infty < t \leq \tau_F \bigg\} \\
&= \bigg\{ X \mapsto
 \frac{\pi(X)}{1 - \pi(X)}
 + 1 - F(t)
 \,:\, -\infty < t \leq \tau_F \bigg\}
\end{align*}
and
\begin{equation*}
\mathcal{Z} = \bigg\{ ((X,Z,\delta),(X',Z',\delta')) \mapsto
 \frac{\zeta(X,Z',\delta',ts(X) + m(X))}{g(X)}
 \,:\, -\infty < t \leq \tau_F \bigg\}.
\end{equation*}
We can conclude that $\mathcal{T}$ is a Euclidean class when we have
shown that $\mathcal{K}$, $\mathcal{S}$ and $\mathcal{Z}$ are each
Euclidean classes. The class $\mathcal{K}$ is Euclidean by Lemma 22 of
Nolan and Pollard (1987) with constant envelope $\|K\|_{\infty}$.

We will now show that the class $\mathcal{S}$ is Euclidean. Let
$\epsilon > 0$. Since $F$ is a continuous distribution function, we
can partition the (infinite length) interval $[-\infty,\,\tau_F]$ into
segments $[t_i,\,t_{i + 1}]$ satisfying $\max_i|F(t_{i + 1}) - F(t_i)|
\leq \epsilon$ by taking an $\epsilon$-net of $[0,\,1]$, consisting of
$O(\epsilon^{-1})$ many links, and using the quantile $F^{-1}$ to
define the corresponding points $t_i$, $i =
1,\ldots,O(\epsilon^{-1})$, that partition the interval
$[-\infty,\,\tau_F]$. Monotonicity of $F$ motivates the following
brackets for a function from $\mathcal{S}$:
\begin{equation*}
\frac{\pi(X)}{1 - \pi(X)} + 1 - F(t_{i + 1})
 \leq \frac{\pi(X)}{1 - \pi(X)} + 1 - F(t)
 \leq \frac{\pi(X)}{1 - \pi(X)} + 1 - F(t_i),
 \qquad t_i \leq t \leq t_{i + 1}.
\end{equation*}
Working with the supremum metric, we find the maximal length of our
brackets is $\max_i|F(t_{i + 1}) - F(t_i)| \leq \epsilon$ as
desired. Therefore, the number of brackets required to cover
$\mathcal{S}$ with respect to the supremum metric,
$N_{[\,]}(\epsilon,\mathcal{S},\|\cdot\|_{\infty})$, is
$O(\epsilon^{-1})$. Hence, there is a constant $C > 0$ such that
$N_{[\,]}(\epsilon,\mathcal{S},\|\cdot\|_{\infty}) \leq
C\epsilon^{-1}$, and it follows that $\mathcal{S}$ is Euclidean with
constant envelope $\|\pi\|_{\infty}/(1 - \|\pi\|_{\infty}) + 1$.

To show that the class $\mathcal{Z}$ is Euclidean, we write
$\mathcal{Z} = \mathcal{Z}_1 - \mathcal{Z}_2$ as a difference of two
classes, where
\begin{equation*}
\mathcal{Z}_1 = \bigg\{ \big((X,Z,\delta),\,(X',Z',\delta')\big)
 \mapsto
 \frac{\delta'\1[Z' \leq ts(X) + m(X)]}{\{1 - M(Z'-\,|\,X)\}g(X)}
 \,:\, -\infty < t \leq \tau_F \bigg\}
\end{equation*}
and $\mathcal{Z}_2$ is equal to
\begin{equation*}
\bigg\{ \big((X,Z,\delta),\,(X',Z',\delta')\big)
 \mapsto
 \int_{-\infty}^{ts(X) + m(X)}\,
 \frac{\1[Z' > u]}{\{1 - M(Z'-\,|\,X)\}^2g(X)}
 \, M^1(du\,|\,X)
 \,:\, -\infty < t \leq \tau_F \bigg\}.
\end{equation*}
We can therefore conclude that $\mathcal{Z}$ is a Euclidean class when
we have shown that both $\mathcal{Z}_1$ and $\mathcal{Z}_2$ are each
Euclidean classes.

Letting $\epsilon > 0$, as before with the class $\mathcal{S}$, we can
partition $[-\infty,\,\tau_F]$ into segments using points $t_i$, $i =
1,\ldots,O(\epsilon^{-2})$, such that
\begin{equation*}
\max_i\sup_{0 \leq x \leq 1} \Big|
 M^1\big(t_{i + 1}s(x) + m(x)\,|\,x\big)
 - M^1\big(t_is(x) + m(x)\,|\,x\big) \Big| \leq
 \Big[1 - \|M\|_{\infty}\Big]^{2}
 \bigg[\inf_{0 \leq x \leq 1} g(x)\bigg]^{2} \epsilon^2.
\end{equation*}
Also similar to the above arguments, monotonicity of
the indicator function motivates the following brackets for a function
from $\mathcal{Z}_1$:
\begin{equation*}
\frac{\delta'\1[Z' \leq t_is(X) + m(X)]}{\{1 - M(Z'-\,|\,X)\}g(X)}
 \leq \frac{\delta'\1[Z' \leq ts(X) + m(X)]}{\{1 - M(Z'-\,|\,X)\}g(X)}
 \leq \frac{\delta'\1[Z' \leq t_{i + 1}s(X) + m(X)]}{\{1 - M(Z'-\,|\,X)\}g(X)},
\end{equation*}
when $t_i \leq t \leq t_{i + 1}$. The squared length of the proposed
brackets in the $L_2(P \otimes P)$--metric satisfies
\begin{align*}
&E\Bigg[ \bigg\{
 \frac{\delta'\1[Z' \leq t_{i + 1}s(X) + m(X)]}{\{1 - M(Z'-\,|\,X)\}g(X)}
 - \frac{\delta'\1[Z' \leq t_is(X) + m(X)]}{\{1 - M(Z'\,|\,X)\}g(X)}
 \bigg\}^2 \Bigg] \\
&\leq \Big[1 - \|M\|_{\infty}\Big]^{-2}
 \bigg[\inf_{0 \leq x \leq 1} g(x)\bigg]^{-2}
 \sup_{0 \leq x \leq 1} \Big|
 M^1\big(t_{i + 1}s(x) + m(x)\,|\,x\big)
 - M^1\big(t_is(x) + m(x)\,|\,x\big) \Big| \\
&\leq \epsilon^2,
 \qquad i = 1,\ldots,O\big(\epsilon^{-2}\big).
\end{align*}
Since $\mathcal{Z}_1$ has the constant envelope $(1 -
\|M\|_{\infty})^{-1}[\inf_{0 \leq x \leq 1} g(x)]^{-1}$, it then
follows for the number of brackets required to cover $\mathcal{Z}_1$,
$N_{[\,]}(\epsilon,\mathcal{Z}_1,L_2(P \otimes P))$, is
$O(\epsilon^{-2})$. Hence, $\mathcal{Z}_1$ is Euclidean. The class
$\mathcal{Z}_2$ can also be shown to be Euclidean by a similar (and
easier) argument, which is omitted. We conclude that the class
$\mathcal{Z}$ is Euclidean as desired.

Therefore, $\mathcal{T}$ is Euclidean and the requirements for
Corollary 7 of Sherman (1994) are satisfied. We can decompose $T_n(t)$
into
\begin{equation*}
\chi_n(t) = \frac1{n(n - 1)a_n} \sum_{j \neq k}
 K\bigg(\frac{X_j - X_k}{a_n}\bigg)
 \frac{1 - Q(ts(X_j) + m(X_j)\,|\,X_j)}{1 - \pi(X_j)}
 \frac{\zeta(X_j, Z_j, \delta_j, ts(X_j) + m(X_j))}{g(X_j)},
\end{equation*}
and a remainder term $T_n(t) - \chi_n(t)$ that is equal to
\begin{equation*}
\frac1{na_n} \frac1{n}\sj
 K(0)\frac{1 - Q(ts(X_j) + m(X_j)\,|\,X_j)}{1 - \pi(X_j)}
 \frac{\zeta(X_j,Z_j,\delta_j,ts(X_j) + m(X_j))}{g(X_j)}
 -\frac1{n}\chi_n(t),
\end{equation*}
with $-\infty < t \leq \tau_F$. Therefore, up to symmetry in the
kernel, $\chi_n(t) - E[\chi_n(t)]$ is a $2^{\text{nd}}$-order
degenerate $U$-process. Note, from the discussion on page 439
of Sherman (1994), the kernel function characterizing the process
$\chi_n(t) - E[\chi_n(t)]$ need not be symmetric in its arguments in
order for the conclusions from Sherman (1994) to hold because the
corresponding $U$-process is given by symmetrizing the kernel
function. We can therefore apply Corollary 7 of Sherman (1994) to see
that both $\|\chi_n - E[\chi_n]\|_{\infty} = \Opn$ and the remainder
satisfies $\|T_n - E[T_n] - \chi_n + E[\chi_n]\|_{\infty} =
\opn$. However, $\chi_n(t) - E[\chi_n(t)]$ is not asymptotically
linear.

To continue, approximate $\chi_n(t) - E[\chi_n(t)]$ by its H\'ajek
projection. For large enough $n$, a function $\eta_{a,t} \in
\mathcal{T}_1$ has the H\'ajek projection $h_{a,t} = h_{1,a,t} +
h_{2,a,t}$, where
\begin{align*}
h_{1,a,t}(X_j,Z_j,\delta_j) &=
 E\big[ \eta_{a,t}((X,Z,\delta),(X_j,Z_j,\delta_j))
 \,|\, (X_j,Z_j,\delta_j) \big] \\
&= \int_{-1}^{1}\,
 \frac{1 - Q(ts(X_j + au) + m(X_j + au)\,|\,X_j + au)}{
 1 - \pi(X_j + au)} \\
&\quad\phantom{= \int_{-1}^{1}\,}\times
 \zeta(X_j + au,Z_j,\delta_j,ts(X_j + au) + m(X_j + au))K(u)\,du
\end{align*}
and
\begin{align*}
h_{2,a,t}(X_j,Z_j,\delta_j) &=
 E\big[ \eta_{a,t}((X_j,Z_j,\delta_j),(X,Z,\delta))
 \,|\, (X_j,Z_j,\delta_j) \big] \\
&= \frac{1 - Q(ts(X_j) + m(X_j)\,|\,X_j)}{1 - \pi(X_j)} \\
&\quad\times
 \int_{-1}^{1}\, E\big[
 \zeta(X_j,Z,\delta,ts(X_j) + m(X_j))\,|\,X_j,\,X = X_j + au \big]
 \frac{g(X_j + au)}{g(X_j)}K(u)\,du.
\end{align*}
The H\'ajek projection of $\chi_n(t) - E[\chi_n(t)]$ is then given by
$n^{-1}\sj h_{a_n,t}(X_j,Z_j,\delta_j) - E[h_{a_n,t}(X,Z,\delta)]$,
where the bandwidth parameter $a_n$ appears in place of $a$.

From Lemma 6 of Sherman (1994), it follows that the class of H\'ajek
projections, i.e.\
\begin{equation*}
\bigg\{ (X,Z,\delta) \mapsto
 h_{a,t}(X,Z,\delta) - E\big[h_{a,t}(X,Z,\delta)\big]
 \,:\, 0 \leq a < 1,\,-\infty < t \leq \tau_F \bigg\},
\end{equation*}
is Euclidean from the fact that $\mathcal{T}$ is Euclidean. We can
therefore apply Corollary 4 (ii) of Sherman (1994) to see that
\begin{equation*}
\sup_{t \in \R} \Big| \chi_n(t) - E\big[\chi_n(t)\big]
 - \avj h_{a_n,t}(X_j,Z_j,\delta_j) + E\big[h_{a_n,t}(X,Z,\delta)\big]
 \Big| = O_P(n^{-1}) = \opn.
\end{equation*}
Define the function $\psi_t = h_{0,t} = h_{1,0,t}$, since $h_{2,0,t}
\equiv 0$ and $E[h_{1,0,t}(X,Z,\delta)] = 0$. If we can show that
\begin{equation} \label{equicon_rem}
\sup_{t \in \R} \bigg| \avj h_{a_n,t}(X_j,Z_j,\delta_j)
 - E\big[h_{a_n,t}(X,Z,\delta)\big]
 - \avj \psi_t(X_j,Z_j,\delta_j) \bigg| = \opn,
\end{equation}
then $T_n(t) - E[T_n(t)]$ is asymptotically linear with influence
function
\begin{equation*}
\psi_t(X,Z,\delta) = \frac{1 - Q(ts(X) + m(X)\,|\,X)}{1 - \pi(X)}
 \eta(X,Z,\delta,ts(X) + m(X)).
\end{equation*}
When $T_n(t) - E[T_n(t)]$ is asymptotically linear we can describe the
weak convergence of $T_n(t) - E[T_n(t)]$ by a mean zero Gaussian
process with known covariance structure.

To complete the argument that $T_n(t) - E[T_n(t)]$ is asymptotically
linear, we need to more closely examine the space of H\'ajek
projections and rewrite
\begin{equation*}
\mathcal{H} = \bigg\{ (X,Z,\delta) \mapsto \int_{-1}^{1}\,
 \Big\{ f_{aw,t}(X,Z,\delta) - E\big[f_{aw,t}(X,Z,\delta)\big] \Big\}
 K(w)\,dw
 \,:\, f_{aw,t} \in \mathcal \mathcal{S}'\otimes\mathcal{Z}' \bigg\},
\end{equation*}
where the classes $\mathcal{S}'$ and $\mathcal{Z}'$ are related to the
classes $\mathcal{S}$ and $\mathcal{Z}$ above with
\begin{equation*}
\mathcal{S}' = \bigg\{ X \mapsto
 \frac{\pi(X + b)}{1 - \pi(X + b)} + 1 - F(t)
 \,:\, -\infty < t \leq \tau_F,\,-1 < b < 1 \bigg\}
\end{equation*}
and $\mathcal{Z}' = \mathcal{Z}_1' - \mathcal{Z}_2'$ with
\begin{equation*}
\mathcal{Z}_1' = \bigg\{ (X,Z,\delta) \mapsto
 \frac{\delta\1[Z \leq ts(X + b) + m(X + b)]}{\{1 - M(Z-\,|\,X + b)\}g(X + b)}
 \,:\, -\infty < t \leq \tau_F,\, -1 < b < 1 \bigg\}
\end{equation*}
and
\begin{align*}
\mathcal{Z}_2' = \bigg\{ (X,Z,\delta) &\mapsto
 \int_{-\infty}^{ts(X + b) + m(X + b)}\,
 \frac{\1[Z > u]}{\{1 - M(Z-\,|\,X + b)\}^2g(X + b)}\,M^{1}(du\,|\,X + b) \\
&\quad \,:\, -\infty < t \leq \tau_F,\,-1 < b < 1 \bigg\}.
\end{align*}
Therefore, the amount of entropy residing in the class $\mathcal{H}$ depends on
the amounts of entropy residing in the classes $\mathcal{S}'$ and
$\mathcal{Z}'$.

Write $\mathcal{S}' = \mathcal{O} + \mathcal{P}$ as a sum of classes,
with $\mathcal{O} = \{X \mapsto \pi(X + b)/\{1 - \pi(X + b)\}\,:\,-1 <
b < 1\}$ and $\mathcal{P} = \{X \mapsto 1 - F(t)\,:\,-\infty < t \leq
\tau_F\}$. Let $\epsilon > 0$ and set $b_i$, $i =
1,\ldots,O(\epsilon^{-1})$, as the grid points for an
$\|\pi'\|_{\infty}^{-1}(1 - \|\pi\|_{\infty})^2\epsilon$-net of
$(-1,\,1)$. Assumption \ref{assumpMandM1} implies that
\begin{equation*}
\sup_{0 \leq x \leq 1} \bigg|
 \frac{\pi(x + b)}{1 - \pi(x + b)}
 - \frac{\pi(x + b_i)}{1 - \pi(x + b_i)} \bigg|
 \leq \epsilon,
\end{equation*} 
whenever $b_i \leq b \leq b_{i + 1}$, $i =
1,\ldots,O(\epsilon^{-1})$. Since $\mathcal{O}$ has constant envelope
$\|\pi\|_{\infty}/(1 - \|\pi\|_{\infty})$, it follows that
$N(\epsilon,\mathcal{O},\|\cdot\|_{\infty}) = O(\epsilon^{-1})$ and
therefore $N_{[\,]}(\epsilon,\mathcal{O},\|\cdot\|_{\infty}) =
O(\epsilon^{-1})$ (see the note in the parentheses near the top of
page 84 of van der Vaart and Wellner, 1996). Repeating the steps above
for showing the class $\mathcal{S}$ satisfies
$N_{[\,]}(\epsilon,\mathcal{S},\|\cdot\|_{\infty}) = O(\epsilon^{-1})$
yields that $N_{[\,]}(\epsilon,\mathcal{P},\|\cdot\|_{\infty}) =
O(\epsilon^{-1})$ as well. Therefore, there is a constant $C > 0$ not
depending on $\epsilon$ such that
\begin{equation} \label{Sprime_brackets}
N_{[\,]}(\epsilon,\mathcal{S}',L_2(P)) \leq C\epsilon^{-2},
 \qquad \epsilon > 0.
\end{equation}

Similarly, write $\mathcal{Z}_1' = \mathcal{I}\otimes\mathcal{D}$ as a
product of classes, with $\mathcal{I} = \{(X,Z,\delta) \mapsto
\delta\1[Z \leq ts(X + b) + m(X + b)]\,:\,-\infty < t \leq \tau_F,\,-1
< b < 1\}$ and $\mathcal{D} = \{(X,Z,\delta) \mapsto \delta/\{\{1 -
M(Z-\,|\,X + b)\}g(X + b)\}\,:\,-1 < b < 1\}$. Now set $\mathcal{M} =
\{X \mapsto m(X + b)\,:\,-1 < b < 1\}$ and $\mathcal{V} = \{X \mapsto
s(X + b)\,:\,-1 < b < 1\}$. Since Assumption \ref{assumpMandM1}
implies that both of $m$ and $s$ are twice differentiable with bounded
derivatives, it is easy to show that
$N_{[\,]}(\epsilon,\mathcal{M},\|\cdot\|_{\infty}) = O(\epsilon^{-1})$
and $N_{[\,]}(\epsilon,\mathcal{V},\|\cdot\|_{\infty}) =
O(\epsilon^{-1})$ for any $\epsilon > 0$. We can therefore choose
brackets $m_i^l < m_i^u$, $i = 1,\ldots,O(\epsilon^{-2})$, for $m$
and brackets $s_j^l < s_j^u$, $j = 1,\ldots,O(\epsilon^{-2})$, for $s$
such that
\begin{equation*}
\|m_i^u - m_i^l\|_{\infty} \leq \frac1{2\|m^{1}\|_{\infty}} \epsilon^2
\quad\text{and}\quad
\|s_j^u - s_j^l\|_{\infty} \leq \frac12
 \bigg[\sup_{0 \leq x \leq 1}\sup_{-\infty < t \leq \tau_0}
 |tm^1(t\,|\,x)|\bigg]^{-1} \epsilon^2.
\end{equation*} 
Proceeding along similar lines as the proof of Lemma A.1 of Van
Keilegom and Akritas (1999) shows that
$N_{[\,]}(\epsilon,\mathcal{I},L_2(P)) = O(\epsilon^{-6})$. It is easy
to show that $\mathcal{D}$ satisfies
$N_{[\,]}(\epsilon,\mathcal{D},L_2(P)) = O(\epsilon^{-2})$. It then
follows that $N_{[\,]}(\epsilon,\mathcal{Z}_1',L_2(P)) =
O(\epsilon^{-8})$.

The class $\mathcal{Z}_2'$ is treated similarly to $\mathcal{Z}_1'$
above, and, with additional technical effort, one shows that
$N_{[\,]}(\epsilon,\mathcal{Z}_2',L_2(P)) =
O(\epsilon^{-6})$. Combining this result with the order for the
bracketing numbers $N_{[\,]}(\epsilon,\mathcal{Z}_1',L_2(P))$ of
$\mathcal{Z}_1'$ above implies that there is a constant $C > 0$ not
depending on $\epsilon$ such that
\begin{equation} \label{Zprime_brackets}
N_{[\,]}(\epsilon,\mathcal{Z}',L_2(P)) \leq C\epsilon^{-14},
 \qquad \epsilon > 0.
\end{equation}

Combining \eqref{Sprime_brackets} and \eqref{Zprime_brackets} shows
the class $\mathcal{S}'\otimes\mathcal{Z}'$ satisfies
$N_{[\,]}(\epsilon,\mathcal{S}'\otimes\mathcal{Z}',L_2(P)) =
O(\epsilon^{-16})$. Since $\mathcal{H}$ has the constant envelope $U =
2\|K\|_{\infty}[\|\pi\|_{\infty}/(1 - \|\pi\|_{\infty}) +
1]\|\eta\|_{\infty}$, we can see that
$N_{[\,]}(\epsilon,\mathcal{H},L_2(P)) = O(\epsilon^{-16})$ as
well. Therefore, only one bracket is required when $\epsilon >
U$. Otherwise, there are constants $C_1,C_2 > 0$ not depending on
$\epsilon$ such that
\begin{equation*}
\int_0^U\,\sqrt{\log N_{[\,]}(\epsilon,\mathcal{H},L_2(P))}\,d\epsilon
 \leq C_1 + C_2\int_0^1\,\sqrt{\log(1/\epsilon)}\,d\epsilon < \infty.
\end{equation*}
It then follows that the class $\mathcal{H}$ is Donsker.

From Corollary 2.3.12 of van der Vaart and Wellner (1996), the class
of empirical processes indexed by the Donsker class $\mathcal{H}$ is
asymptotically equicontinuous in the sense that, for any $\epsilon >
0$,
\begin{equation} \label{equicon}
\lim_{\alpha \downarrow 0} \limsup_{n \to \infty} P \Bigg(
 \sup_{h_1,h_2 \in \mathcal{H}\,:\,\Var(h_1 - h_2) < \alpha} \bigg|
 n^{-1/2}\sj \Big\{ h_1(X_j,Z_j,\delta_j) - h_2(X_j,Z_j,\delta_j) \Big\}
 \bigg| > \epsilon \Bigg) = 0.
\end{equation}
We can see that \eqref{equicon} implies the desired
\eqref{equicon_rem} if we can show that $h_{a_n,t} - h_{0,t}$
satisfies the variation condition under the norm inside the
probability statement in \eqref{equicon}, where the norm inside the
probability statement is restricted to the subclass of functions from
$\mathcal{H}$ with $\{a_n\}_{n \geq 1}$ in place of $0 \leq a < 1$.

Write $h_{a_n,t} - h_{0,t} = h_{1,a_n,t} - h_{1,0,t}
+ h_{2,a_n,t}$, and observe that
\begin{align*}
\big|h_{2,a_n,t}(X_j,Z_j,\delta_j)\big| &=
 \frac{1 - Q(ts(X_j) + m(X_j)\,|\,X_j)}{1 - \pi(X_j)} \\
&\quad\times
 \bigg| \int_{-1}^{1}\,
 E\Big[\eta\big(X_j,Z,\delta,ts(X_j) + m(X_j)\big)\,\Big|\,X_j,X = a_nu\Big]
 \frac{g(X_j + a_nu)}{g(X_j)}K(u)\,du \bigg| \\
&\leq 2\|K\|_{\infty}\big(1 - \|\pi\|_{\infty}\big)^{-1}
 \frac{\|g\|_{\infty}}{\inf_{0 \leq x \leq 1} g(x)} \\
&\quad\times \sup_{-\infty < t \leq \tau_F}\sup_{-1 < u < 1} \Big|
 E\Big[\eta\big(X_j,Z,\delta,ts(X_j) + m(X_j)\big)
 \,\Big|\,X_j,X = X_j + a_nu\Big] \Big| \\
&= O(a_n),
 \qquad\text{a.s.},
\end{align*}
which follows from the facts that $E[\eta(X_j,Z,\delta,ts(X_j) +
m(X_j))\,|\,X_j,X = x]$ is both bounded and differentiable in $x$
and that $E[\eta(X_j,Z,\delta,ts(X_j) + m(X_j))\,|\,X_j] = 0$. The
variance of $h_{a_n,t} - h_{1,0,t}$ satisfies
\begin{align*}
&E\bigg[ \Big\{ h_{1,a_n,t}(X,Z,\delta) - h_{1,0,t}(X,Z,\delta)
 + h_{2,a_n,t}(X,Z,\delta) - E\big[h_{1,a_n,t}(X,Z,\delta)\big]
 - E\big[h_{2,a_n,t}(X,Z,\delta)\big] \Big\}^2 \bigg] \\
&\leq 2 E\bigg[ \Big\{
 h_{1,a_n,t}(X,Z,\delta) - h_{1,0,t}(X,Z,\delta) \Big\}^2 \bigg]
 + 4 E^2\Big[ h_{1,a_n,t}(X,Z,\delta) \Big] \\
&\quad + 8 E\Big[h_{2,a_n,t}^2(X,Z,\delta)\Big]
 + 16 E^2\Big[ h_{2,a_n,t}(X,Z,\delta) \Big] \\
&= 2 E\bigg[ \Big\{
 h_{1,a_n,t}(X,Z,\delta) - h_{1,0,t}(X,Z,\delta) \Big\}^2 \bigg]
 + 4 E^2\Big[ h_{1,a_n,t}(X,Z,\delta) \Big]
 + O\big(a_n^2\big),
\end{align*}
and the quantity $E[\{h_{1,a_n,t}(X,Z,\delta) -
h_{1,0,t}(X,Z,\delta)\}^2]$ is equal to
\begin{align*}
&E\Bigg[ \bigg\{
 \int_{-1}^{1} \bigg\{
 \frac{1 - Q(ts(X + a_nu) + m(X + a_nu)\,|\,X + a_nu)}{1 - \pi(X + a_nu)} \\
&\phantom{E\Bigg[ \bigg\{ \int_{-1}^{1} \bigg\{}
 \quad\times
 \zeta\big(X + a_nu,Z,\delta,ts(X + a_nu) + m(X + a_nu)\big) \\
&\phantom{E\Bigg[ \bigg\{ \int_{-1}^{1} \bigg\{}
 \quad - \frac{1 - Q(ts(X) + m(X)\,|\,X)}{1 - \pi(X)}
 \zeta\big(X,Z,\delta,ts(X) + m(X)\big) \bigg\}K(u)\,du \bigg\}^2 \Bigg] \\
&= O(a_n),
 \qquad -\infty < t \leq \tau_F,
\end{align*}
which follows from the facts that both
\begin{align*}
&\sup_{0 \leq x \leq 1} \sup_{-1 < u < 1} \bigg|
 \frac{1 - Q(ts(x + a_nu) + m(x + a_nu)\,|\,x + a_nu)}{1 - \pi(x + a_nu)}
 - \frac{1 - Q(ts(x) + m(x)\,|\,x)}{1 - \pi(x)} \bigg| \\
&= \sup_{0 \leq x \leq 1} \sup_{-1 < u < 1} \bigg|
 \frac{\pi(x + a_nu)}{1 - \pi(x + a_nu)}
 - \frac{\pi(x)}{1 - \pi(x)} \bigg| \\
&= O(a_n)
\end{align*}
and
\begin{equation*}
\sup_{-\infty < t \leq \tau_F} \sup_{-1 < u < 1} \bigg|
 E\bigg[ \Big\{
 \zeta\big(X + a_nu,Z,\delta,ts(X + a_nu) + m(X + a_nu)\big)
 - \zeta\big(X,Z,\delta,ts(X) + m(X)\big) \Big\}^2 \bigg]
\end{equation*}
is of the order $O(a_n)$. Similarly, conclude that
$E^2[h_{1,a_n,t}(X,Z,\delta)] = O(a_n^2)$, uniformly in $-\infty < t \leq
\tau_F$. Therefore, the variance of $h_{a_n,t} - h_{0,t}$ is
asymptotically negligible, and we can apply \eqref{equicon} to obtain
the desired \eqref{equicon_rem}. We conclude that $T_n(t) - E[T_n(t)]$
is asymptotically linear with influence function
\begin{equation*}
\psi_t(X,Z,\delta) =
 \frac{1 - Q(ts(X) + m(X)\,|\,X)}{1 - \pi(X)}
 \zeta\big(X,Z,\delta,ts(X) + m(X)\big),
 \qquad -\infty < t \leq \tau_F.
\end{equation*}
It follows that the process $\{T_n(t) - E[T_n(t)]\,:\,-\infty < t \leq
\tau_F\}$ weakly converges to a mean zero Gaussian process
$\{Z_T(t)\,:\,-\infty < t \leq \tau_F\}$ with covariance function
$\Sigma_T(t,v) = E[\psi_t(X,Z,\delta)\psi_v(X,Z,\delta)]$ for $-\infty
< t,v \leq \tau_F$. The assertion then follows by finding similar
conclusions for the random quantities $U_n - E[U_n]$, $V_n - E[V_n]$
and $W_n(t) - E[W_n(t)]$.
\end{proof}

\section*{Acknowledgements}
The authors would like to thank and acknowledge the following sources
of financial support. This research has been supported by the European
Research Council (2016-2021, Horizon 2020 / ERC grant agreement No.\
694409), the IAP research network grant nr.\ P7/06 of the Belgian
government (Belgian science policy), the Collaborative Research Center
``Statistical modelling of nonlinear dynamic processes.'' (SFB 823,
Project C4) of the German Research Foundation.




\begin{thebibliography}{41}

\bibitem{A1989}
Aitkin, M.\ (1989).
{\em Statistical modelling in GLIM}.
Oxford science publications. Clarendon Press, Oxford.

\bibitem{AVK2001}
Akritas, M.G.\ and Van Keilegom, I.\ (2001).
Non-parametric estimation of the residual distribution.
{\em Scand. J. Statist.} {\bf 28}, 549-567.

\bibitem{AG1993}
Arcones, M.A.\ and Gin\'e, E.\ (1993).
Limit theorems for $U$-processes.
{\em Ann. Probab.} {\bf 21}, 1494-1542.

\bibitem{B1981}
Beran, R.\ (1981).
Nonparametric regression with randomly censored survival data.
Tech. Rep., University of California, Berkley.

\bibitem{B1949}
Boag, J.W.\ (1949).
Maximum likelihood estimates of the proportion of patients cured by
cancer therapy.
{\em J. R. Stat. Soc. Ser. B Stat. Methodol.} {\bf 11}, 15-53.

\bibitem{CJY2002}
Chen, K., Jin, Z.\ and Ying, Z.\ (2002).
Semiparametric analysis of transformation models with censored data.
{\em Biometrika} {\bf 89}, 659-668.

\bibitem{CIS1999}
Chen, M.H., Ibrahim, J.G.\ and Sinha, D.\ (1999).
A new Bayesian model for survival data with a surviving fraction.
{\em J. Amer. Statist. Assoc.} {\bf 94}, 909-919.

\bibitem{C2016}
Chown, J.\ (2016).
Efficient estimation of the error distribution function in
heteroskedastic nonparametric regression with missing data.
{\em Statist. Probab. Lett.} {\bf 117}, 31-39.

\bibitem{C1994}
Collett, D.\ (1994).
{\em Modelling survival data in medical research.}
CRC monographs on statistics \& applied probability.
Taylor \& Francis, Oxfordshire.

\bibitem{D1987}
Dabrowska, D.M.\ (1987).
Non-parametric regression with censored survival time data.
{\em Scand. J. Statist.} {\bf 14}, 181-197.

\bibitem{DA2002}
Du, Y.\ and Akritas, M.G.\ (2002).
Uniform strong representation of the conditional Kaplan-Meier process.
{\em Math. Methods Statist.} {\bf 11}, 152-182.

\bibitem{E1981}
Efron, B.\ (1981).
Censored data and the bootstrap.
{\em J. Amer. Statist. Assoc.} {\bf 76}, 312-319.

\bibitem{F1986}
Farewell, V.T.\ (1986).
Mixture cure models in survival analysis: Are they worth the risk?
{\em Canad. J. Statist.} {\bf 14}, 257-262.

\bibitem{HA1991}
Harris, E.K.\ and Albert, A.\ (1991).
{\em Survivorship analysis for clinical studies.}
Statistics: A series of textbooks and monographs.
Taylor \& Francis, Oxfordshire.

\bibitem{H1959}
Haybittle, J.L.\ (1959).
The estimation of the proportion of patients cured after treatment for
cancer of the breast.
{\em Br. J. Radiol.} {\bf 32}, 725-733.

\bibitem{H1965}
Haybittle, J.L.\ (1965).
A two-parameter model for the survival curve of treated cancer
patients.
{\em J. Amer. Statist. Assoc.} {\bf 60}, 16-26.

\bibitem{KC1992}
Kuk, Y.C.\ and Chen, C.H.\ (1992).
A mixture model combining logistic regression with proportional
hazards regression.
{\em Biometrika} {\bf 79}, 531-541.

\bibitem{L1982}
Lawless, J.F.\ (1982).
{\em Statistical models and methods for lifetime data.}
Wiley series in probability and mathematical statistics: Applied
probability and statistics. Wiley, New York.

\bibitem{LCCAJVK2017}
L\'opez-Cheda, A., Cao, R., Amalia J\'acome, M.\ and Van Keilegom, I.\
(2017).
Nonparametric incidence estimation and bootstrap bandwidth selection
in mixture cure models.
{\em Comput. Statist. Data Anal.} {\bf 105}, 144-165.

\bibitem{L2008}
Lu, W.\ (2008).
Maximum likelihood estimation in the proportional hazards cure model.
{\em Ann. Inst. Stat. Math} {\bf 60}, 545-574.

\bibitem{L2010}
Lu, W.\ (2010).
Efficient estimation for an accelerated failure time model with a cure
fraction.
{\em Stat. Sin.} {\bf 20}, 661-674.

\bibitem{NP1987}
Nolan, D.\ and Pollard, D.\ (1987).
$U$-processes: Rates of convergence.
{\em Ann. Statist.} {\bf 15}, 780-799.

\bibitem{PP1989}
Pakes, A.\ and Pollard, D.\ (1989).
Simulation and the asymptotics of optimization estimators.
{\em Econometrica} {\bf 57}, 1027-1057.

\bibitem{PVK2017}
Patilea, V.\ and Van Keilegom, I.\ (2017).
A general approach for cure models in survival analysis.
Submitted.

\bibitem{PVKEG2017}
Portier, F., Van Keilegom, I.\ and El Ghouch, A.\ (2017).
On an extension of the promotion time cure model.
Submitted.

\bibitem{S1994}
Sherman, R.P.\ (1994).
Maximal inequalities for degenerate $U$-processes with applications to
optimization estimators.
{\em Ann. Statist.} {\bf 22}, 439-459.

\bibitem{SCI2003}
Sinha, D., Chen, MH.\ and Ibrahim, J.G.\ (2003).
{\em Bayesian inference for survival data with a surviving fraction.}
Crossing boundaries: statistical essays in honor of Jack Hall.
Institute of Mathematical Statistics, Beachwood, Ohio.

\bibitem{S1977}
Stone, C.J.\ (1977).
Consistent nonparametric regression.
{\em Ann. Statist.} {\bf 5}, 595-620.

\bibitem{ST2000}
Sy, J.P.\ and Taylor, J.M.G.\ (2000).
Estimation in a Cox proportional hazards cure model.
{\em Biometrics} {\bf 56}, 227-236.

\bibitem{T1995}
Taylor, J.M.G.\ (1995).
Semi-parametric estimation in failure time mixture models.
{\em Biometrics} {\bf 51}, 899-907.

\bibitem{T1998}
Tsodikov, A.\ (1998).
A proportional hazards model taking account of long-term survivors.
{\em Biometrics} {\bf 54}, 1508-1516.

\bibitem{T2003}
Tsodikov, A.\ (2003).
Semiparametric models: A generalized self-consistency approach.
{\em J. R. Stat. Soc. Ser. B Stat. Methodol.} {\bf 65}, 759-774.

\bibitem{TIY2003}
Tsodikov, A., Ibrahim, J.\ and Yakovlev, A.\ (2003).
Estimating cure rates from survival data: an alternative to
two-component mixture models.
{\em J. Am. Stat. Assoc.} {\bf 98}, 1063-1078.

\bibitem{VW1998}
van der Vaart, A.W.\ and Wellner J.A.\ (1996).
{\em Weak convergence and empirical processes. With applications to 
statistics.}
Springer Series in Statistics. Springer-Verlag, New York.

\bibitem{VKA1999}
Van Keilegom, I.\ and Akritas, M.G.\ (1999)
Transfer of tail information in censored regression models.
{\em Ann. Statist.} {\bf 27}, 1745-1784.

\bibitem{Wetal2005}
Wang, Y., Klijn, J.G.M., Zhang, Y., Sieuwerts, A.M., Look, M.P., Yang,
F., Talantov, D., Timmermans, M., Meijer-van Gelder, M.E., Yu, J.,
Jatkoe, T., Berns, E.M.J.J., Atkins, D. and Foekens, J.A.\ (2005).
Gene-expression profiles to predict distant metastasis of
lymph-node-negative primary breast cancer.
{\em Lancet} {\bf 365}, 671-679.

\bibitem{XP2014}
Xu, J.\ and Peng, Y.\ (2014).
Nonparametric cure rate estimation with covariates.
{\em Canad. J. Statist.} {\bf 42}, 1-17.

\bibitem{YCS1994}
Yakovlev, A.Y., Cantor, A.B.\ and Shuster, J.J.\ (1994).
Parametric versus nonparametric methods for estimating cure rates
based on censured survival data.
{\em Stat. Med.} {\bf 13}, 983-986.

\bibitem{YT1996}
Yakovlev, A.Y.\ and Tsodikov, A.\ (1996).
{\em Stochastic models of tumor latency and their biostatistical
  applications.}
Series in mathematical biology and medicine.
World Scientific Pub.\ Co., Singapore.

\bibitem{YI2005}
Yin, G.\ and Ibrahim, J.G.\ (2005).
Cure rate models: a unified approach.
{\em Canad. J. Statist.} {\bf 33}, 559-570.

\bibitem{ZYI2006}
Zeng, D., Yin, G.\ and Ibrahim, J.G.\ (2006).
Semiparametric transformation models for survival data with a cure
fraction.
{\em J. Amer. Statist. Assoc.} {\bf 101}, 670-684.

\end{thebibliography}
\end{document}